\documentclass{rQUF2e}

\usepackage{epstopdf}% To incorporate .eps illustrations using PDFLaTeX, etc.
\usepackage{subfigure}% Support for small, `sub' figures and tables
\usepackage[normalem]{ulem}
\theoremstyle{plain}
\newtheorem{theorem}{Theorem}[section]
\newtheorem{corollary}[theorem]{Corollary}

\theoremstyle{definition}
\newtheorem{definition}{Definition}

\theoremstyle{remark}
\newtheorem{remark}{Remark}

\usepackage{amscd}
\usepackage{amssymb}
\usepackage{amsmath}
\usepackage{amsthm}
\usepackage{amsfonts}

\usepackage{latexsym}

\usepackage{url}

\usepackage[utf8]{inputenc}
\usepackage[T1]{fontenc}
\usepackage{ulem}

\usepackage{graphicx}
\usepackage{multirow}

\usepackage{enumerate}

\usepackage{bm}

\usepackage{color}

\usepackage{fancyvrb}

\usepackage{verbatim}
\usepackage{bbm}

\usepackage{tabto}

\usepackage{framed}

\usepackage{multirow}
\usepackage{booktabs}

\usepackage{lipsum}

\usepackage{tikz}
\usetikzlibrary{shapes.geometric,arrows,positioning}

\usepackage[dvips,ps2pdf]{hyperref}
\hypersetup{
    colorlinks = true,
    citecolor = blue,
    linkcolor = red,
    pagebackref = true,
    breaklinks = true,
}
\usepackage{breakurl}

\usepackage[algosection]{algorithm2e}

\begin{document}

%\jvol{00} \jnum{00} \jyear{2014} \jmonth{October}

\title{{\textit{Myopic robust index tracking with Bregman divergence}}}
\bigskip

\author{S. PENEV$^{\ast}$$\dag$\thanks{$^\ast$Corresponding author.
Email: S.Penev@unsw.edu.au}  P.V. SHEVCHENKO${\ddag}$ W. WU$\dag$\\
\affil{$\dag$School of Mathematics and Statistics, UNSW Sydney, NSW 2052 Australia\\
%$\ddag$Department of Actuarial Studies and Business Analytics, Macquarie University, Sydney, NSW 2109 Australia}
$\ddag$Department of Actuarial Studies and Business Analytics, Macquarie University, Sydney, NSW 2109 Australia, and Center for Econometrics and Business Analytics, St. Petersburg State University, Saint Petersburg, Russia
}
\received{final version  released June 2021} }
\bigskip
\maketitle

\begin{abstract}
Index tracking is a popular form of asset management. Typically, a quadratic function is used to define the tracking error of a portfolio and the look back approach is applied to solve the index tracking problem. We argue that a forward looking approach is more suitable, whereby the tracking error is expressed as an expectation of a function of the difference between the returns of the index and of the portfolio. We also assume that there is a model uncertainty in the distribution of the assets, hence a robust version of the optimization problem needs to be adopted. We use Bregman divergence in describing the deviation between the nominal and actual (true) distribution of the components of the index. In this scenario, we derive the optimal robust index tracking portfolio  in a semi-analytical form as
a solution of a system of nonlinear equations. Several numerical results are presented that allow us to compare the performance of this robust portfolio with
the optimal non-robust portfolio. We show that, especially during market downturns, the robust portfolio can be very advantageous.
\end{abstract}
\bigskip
\begin{keywords}
Index tracking; Robust index tracking; Bregman divergence; Kullback-Leibler divergence
\end{keywords}
\bigskip

\begin{classcode}G11; D81\end{classcode}
\newpage

\section{Introduction.} \label{intro}

A popular form of passive asset management is the so-called index tracking  (see discussions, for example, in \cite{AndN14,BeaMC03,GaiKW05}). Essentially, it means that the fund
manager (or the investor) tries to replicate the performance of an index either through its value or its return (see \cite{StrB18} and the references therein).
In a frictionless and liquid market, a full replication portfolio (i.e. by holding exactly the same composition as the index) obviously yields the best
tracking performance. This  has already been discussed in many past studies, e.g., \cite{BeaMC03,StrB18}. However,  if transaction costs are
considered or some of the components of the index are illiquid (see \cite{MagTMP07}), then a full replication portfolio does not necessarily deliver the best performance.
 This is due to the fact that a full replication portfolio will involve high transaction costs and because buying and selling  illiquid assets will be difficult. Thus, to replicate the index, the fund manager may
choose a tracking portfolio with only a subset of representative assets (see \cite{StrB18, dePdC16, GuaS12} among others). It is worth noting that, in
general, the assets in the tracking portfolio do not have to be the components of the index as long as they exhibit  good similarity with the index (see for
example, \cite{AndN14}). \\

To  satisfactorily solve the index tracking problem or the related enhanced tracking problem (which tracks the index as well as outperforms it), predominantly the look back approach has
been used in the literature, see e.g., \cite{StrB18,GuaMOS16,ChiTM13,MonVFD08,BeaMC03}. This approach relies on the assumption that a portfolio that has tracked
the index well in the past will also demonstrate good tracking performance in the future. The past realizations of
the return (or value) of the index and the return (or value) of the tracking portfolio are collected and the tracking error is defined as a function of the
difference between  the index and the tracking portfolio. A quadratic function is often used to define such a tracking error.\\

The above approach may lead to  poor performance if the future differs vastly from the past. Another way to solve the index tracking problem is by adopting the  forward looking approach. This is based on defining the tracking error as the expectation of a function of the difference between the return of the index and the
return of the tracking portfolio (see for example, \cite{dePdC16,MeaS90}). The expectation is calculated by using the joint distribution  of the index
and of the tracking portfolio. A reliable estimation of the joint distribution is required in order to  guarantee a good performance of this approach. If uncertainty in the
distribution of the assets is present, a robust version of the approach needs to be taken.
%This motivates our current work.
\\

%{\color{red}
 From the onset of the paper, we should point out that index tracking has not received the attention it deserves until recently, but the situation has since changed dramatically. Many aspects of index tracking attract the attention of the research community.\\

 An interesting recent research is directed towards the attempt to unify the two tasks, selecting the ``optimal subset'' and determining the ``best'' weights in a single step. This approach is attractive especially if an index contains a large number of stocks while one is interested  to have a (very) small number of stocks in the portfolio (\cite{Benidis18}). In other words, the goal in this case is to select a sparse index tracking portfolio. While we  are also  interested in developing such types of  approaches, they are based on different optimization techniques (e.g., sparse non-convex optimization \cite{Benidis18}, mixed-integer programming \cite{Canakgoz2008}, \cite{Filippi2016}), and  we are not discussing these approaches in our paper. \\

 It may be possible to use quite  sophisticated statistical multivariate models and numerical approximation techniques to better mimic the underlying joint distribution of the stock returns and the index returns, and our theory allows for such models to be used as a nominal distribution in index tracking. However, our  point of view in this paper is that even then a robustness approach is worth using as it has the potential to safeguard against deviations from the nominal model.\\
  %end red}

%{\color{red}
The paper \cite{Gnagi2020} gives a review of recent approaches  and examines critically the gaps in the literature on enhanced index-tracking problem. It also proposes novel matheuristics (combining metaheuristics with mathematical programming) numerical algorithms to solve index-tracking problems with very large indices (e.g., with more than 9000 constituents).   Among the many constraints they consider, is also the achievement of a given target excess return in the future. Because future outcomes are uncertain, they end up  requiring a minimum expected excess return which may not deliver the target in an out-of-sample period. The paper does not deal with the robustness-type penalties that are of main interest in our paper.
%end}
\\

%{\color{red}
As mentioned, the most recent literature focuses on the look back approach and a great effort has been made to model transaction costs and other sophisticated
restrictions on the tracking portfolio such as choosing which components of the index to include. The aim of this paper is,
however, to develop a {\it robust}  version of index tracking. As far as the authors are aware, this has not been
done in the past. When done,the concern has mainly been about robustness with respect to parameter uncertainty rather than robustness about distributional uncertainty (\cite{Fabozzi2007}). One possible  exception is  \cite{Lej12} where a minimax game theoretic interpretation of the robust forward looking approach is discussed. However, in \cite{Lej12},
only parameter uncertainty is considered, that is, the authors assume that the excess returns are imperfectly known but belong to a  class of distributions characterized by an ellipsoidal distributional set. In contrast, we consider the uncertainty in the joint distribution of the index and the tracking portfolio, and find the
optimal way to track the index under the worst case distribution. Our approach could be compared to the recently popular distributional robustness methodology. A survey of this methodology, with a focus on actuarial applications, can be found in \cite{BLTY2019}. To the best of our knowledge, we have not seen application of this methodology to the index tracking problem.\\

The model uncertainty in our paper is measured through a special form of Bregman divergence (see \cite{Bre67}). The notion of Bregman divergence is quite general. It has been used as a means to measure the pairwise dissimilarity between matrices (\cite{PenPr16}), between vectors (\cite{BMDG05}), and also between functions (\cite{GD14}, \cite{PenN18}). In the latter case, the authors in \cite{GD14} call it a functional Bregman divergence. Precise definition is given in Section \ref{myopic}.\\

The classical Kullback-Leibler (KL) divergence also belongs to the class of  Bregman divergences. It can be used as a benchmark. However, KL divergence  is not appropriate to handle heavy tailed distributions which are commonly  observed in financial asset returns (\cite{DeyJ10}, \cite{PocS11}). Hence we choose another family of Bregman divergences whereby the convex function in the divergence's definition in Section \ref{myopic} has a stronger polynomial growth than the one that is used in the case of KL.  Yet we are using a family of Bregman divergences that is parameterized by one positive parameter
%{\color{red}
{$\lambda$}
%end red }
 only and is such that allows us to recover the KL divergence in a limit when
%{\color{red}
{$\lambda \to 0$}.
%end red}
  In such a way we can study the effect of stronger robustification achieved when
$\lambda $ runs away from the zero value. As we point out in detail in Remark 3.3, this goal is indeed achieved by our choice of the function $F_{\lambda }(.)$ in the specification of the Bregman divergence. \\

Our first contribution is the derivation of a semi-closed form of the worst case distribution
%{\color{red}
 for
 %end red}
  the chosen  Bregman divergence.
 %{\color{red}
  The second contribution is the derivation of the optimal index tracking portfolio in a semi-analytical form.  The derivation of such a robust index tracking portfolio was the main goal of our work.  Thirdly, we investigate the performance of our  proposed robust tracking portfolio
in a short numerical study aiming to demonstrate the robustness effect. Our next contribution is related to extending our approach to deal with a variety
of loss functions that are suitable to measure the quality of an index tracking portfolio. Finally, we also include a real data example in order to illustrate the
out-of-sample performance of our method.\\
%end red}

The structure of this paper is outlined below. In Section 2, we formulate the index tracking problem, and present the look back approach and the forward looking
approach through a simple example. In Section 3, we formulate the robust index tracking problem and derive the robust index tracking portfolio. In Section 4, we extend our model to tackle enhanced index tracking. In Section 5 we  present our numerical study. Section 6 concludes  and  outlines  avenues for further research. \\

% In section 4, we made some discussions on how to add transaction costs to our model. Several numerical examples are presented in section 5. In seciton 6, we
% conclude the paper.

% The first stage involves choosing the number of assets which shall be included in the tracking portfolio either through a subjective way
% (\cite{dePdC16}) or as pointed out in \cite{GuaS12}, through some statistical test (see \cite{DosC05}). The minimization of tracking error is then carring
% out in the second stage.   \\

\section{Myopic Index Tracking} \label{sec:probform}

In this section, we formulate the index tracking problem and compare the look back approach and the forward looking approach through a simple example.
%\subsection{The Index Tracking Problem}
Consider an index  that consists of $\ell > 1$ risky assets. A fund manager is interested in constructing a tracking portfolio that contains $d < \ell$ risky assets
that may not necessarily belong to the index.
The aim  is to replicate the return of the index over a fixed investment period.
 In particular, since the return of the index is available for analysis  after each investment period, we focus on a short term one-period (called ``myopic'') index tracking.\\

 Within this period, we denote by $\bm{r}$  the random vector of returns of the
risky assets  $\bm{r} = (r^{1}, ..., r^{d})^\intercal $  included in the portfolio. Hence  $r^{i}$, $i = 1, ...,d$ denotes the return of the $i$th individual asset over the intended investment period. The return of the index over this investment period is denoted as $b$. We then define  $\bm{R} = \bold{1}+\bm{r}$, and $B = 1 + b$. \\

Throughout the paper, we  assume that all random quantities are  defined on a
complete probability space ($\Omega, \mathcal{F}, \mathbb{P}$) with the sample space $\Omega$, the $\sigma$-algebra $\mathcal{F}$, and the probability measure
$\mathbb{P}$, where the $\sigma$-algebra $\mathcal{F} =  \sigma(\bm{r},B)$.

In addition, we assume that short selling is permitted.\\

At the beginning of the investment period, the fund manager re-balances the portfolio with a control  $\bm{u}$, where $\bm{u}^{i} \in \mathbb{R}$, $i = 1, ..., d$,
denotes the proportional allocation of the wealth of the investor into the $i$th asset. Let $\mathcal{U}$ be the set of admissible strategies $\bm{u} \in \mathcal{U}$ such that
\begin{eqnarray}
  U = \Big\{\bm{u} \in \mathbb{R}^{d}: \bm{1}^\intercal\bm{u} = 1\Big\}. \nonumber
\end{eqnarray}

% \begin{enumerate}[leftmargin]
%  \item[\textbf{(A1).}] All random quantities in this paper will be assumed to be defined on a complete probability space ($\Omega, \mathcal{F}, \mathbb{P}$) with the
%                        sample space $\Omega$, the $\sigma$-algebra $\mathcal{F}$, and the probability measure $\mathbb{P}$, where the $\sigma$-algebra $\mathcal{F} =
%                        \sigma(\bm{r},B)$.
%  \item[\textbf{(A2).}] Short selling is permitted.
% \end{enumerate}
%The overall goal is to find suitable weights $\bm{u}$ to minimize the deviation between the index and the portfolio.

%\subsection{Two Approaches to Track an Index}
Two approaches exist to track an index.
The look back approach finds the optimal portfolio based on the historical data. Let $\bm{r}_{-n}$ and $b_{-n}, n = T, (T-1), ..., 1$, denote the return of the
risky assets and the return of the index, respectively, on the $n$th day before today. Define $\bm{R}_{-n} = \bm{1} + \bm{r}_{-n}$, and $B_{-n} = 1 + b_{-n}$. Under this approach, the optimal control  $\bm{u}$ can be obtained by solving
\begin{eqnarray}\label{eqn:nonrobustlookback}
  \mathcal{V} = \inf_{\bm{u} \in \mathcal{U}}\frac{1}{T}\sum_{n = 1}^{T}  \big(\bm{R}_{-n}^{\intercal}\bm{u} - B_{-n}\big)^{2}.
\end{eqnarray}
The  joint distribution of the index and of the tracking portfolio  is taken to be the empirical distribution. \\

In contrast, the forward looking approach finds the optimal portfolio by solving % , and $P_{n} = x$ and $I_{n} = y$ are given
\begin{eqnarray}\label{eqn:nonrobust}
  \mathcal{V} &  =  & \inf_{\bm{u} \in \mathcal{U}}\mathbb{E}\Big( \big(\bm{R}^{\intercal}\bm{u} - B\big)^{2}\Big).
\end{eqnarray}
Under this approach, the actual  distribution can be assumed to be essentially arbitrary. \\

It is easy to see that the forward looking approach relies on the future estimation of the actual distribution. If the empirical distribution delivers a good estimate of the future,
then the look back approach is equivalent to the forward looking approach. However, this is not the case if another assumption is made about the actual distribution. Thus, unless
the empirical distribution represents a reliable estimate, the two approaches yield different outcomes  in general. In addition, even though the empirical
distribution may represent a good estimate of the future, there is always an uncertainty in the estimation of the actual distribution.
It is obvious that the  “look back” is just a specific way to perform “forward looking” by using an empirical distribution of risk factors.

\section{Myopic Robust Index Tracking}
\label{myopic}

To model the uncertainty associated with the estimation of the actual distribution, we use the Bregman divergence.

\subsection{The Bregman Divergence}

The following definition of a  Bregman divergence is taken from \cite{PenN18}.
\begin{definition}
Given a strictly convex function $F: A \rightarrow \mathbb{R}$, where $A \subset \mathbb{R}^{d}$ is a convex set, the local Bregman divergence between two
 points $X\in A$ and $Y\in A$ is defined as
 \begin{eqnarray}
  d_{F}(X,Y)  &  =  & F(Y) - F(X)  - \nabla F(X)(Y - X).  \nonumber
\end{eqnarray}
\end{definition}
The above definition can also be applied point wise for
 positive density functions $f$, $g$ defined on a common domain. The point wise application means that in this case $d=1,$  $\nabla $ means a simple derivative $F'$ and we interpret locally, for a fixed $t$
\begin{eqnarray}
  d_{F}(f(t),g(t))  &  =  & F(g(t)) - F(f(t))  - F'(f(t))(g(t) - f(t)).  \nonumber
\end{eqnarray}
Using this localized divergence measure at the point $t,$ we then define the global (or also called {\it functional) Bregman divergence} between the densities $f$ and $g  :$

\begin{eqnarray}
\label{FBD}
   D_{Breg}(f,g)  & :=  &   \displaystyle\int d_{F}(f(t),g(t)) v(t)dt,
\end{eqnarray}
where $v$ is some non-negative weight function.

In the definition of the Bregman divergence, any  strictly convex function $F(\cdot)$ could be chosen. A specific choice has been suggested in the  paper  \cite{PenN18}, which we also take on board  here. We take the strictly convex function
$ F_{\lambda}(z): (0,\infty) \to \mathbb{R}^1$ to be
\begin{eqnarray}
  F_{\lambda}(z) = \frac{z^{\lambda+1} - (\lambda+1)z}{\lambda},   \ \ \textnormal{for a fixed} \ \lambda > 0.       \nonumber
\end{eqnarray}
Then, for two densities $f$ and $g$ of the $d$-dimensional argument $x, $ it is easy to see that
\begin{eqnarray}
  d_{F}(f(x),g(x))  &   =   &   F_{\lambda}(g(x)) - F_{\lambda}(f(x))  - <\nabla F_{\lambda}(f(x)){^{T}}, g(x) - f(x)>  \nonumber \\
%              &   =   &   \frac{g^{\lambda+1}(x) - (\lambda+1)g(x)}{\lambda} - \Big(\frac{f(x)^{\lambda+1} - (\lambda+1)f(x)}{\lambda}\Big)
%                        - (g(x)-f(x))\Big(\frac{(\lambda+1)f(x)^{\lambda} - (\lambda+1)}{\lambda}\Big)  \nonumber \\
              &   =   &   \frac{g(x)^{\lambda+1}}{\lambda} - \frac{(\lambda+1)g(x)f(x)^{\lambda}}{\lambda} + f(x)^{\lambda+1}.  \nonumber
\end{eqnarray}

It is worth noting that our special form of Bregman divergence is closely related to the so-called Tsallis divergence and $\alpha$-divergence (see for example \cite{CicA10,PocS11}).
%{\color{red}
 $\alpha$-divergence is closely related to the special case of our Bregman divergence. It differs by a re-parameterization on which both divergences are defined.
It is quite a laborious task to discuss and compare the many different types of divergences that exist in the literature. Such a comparison is done in the monograph  \cite{Amari16} and in \cite{DeyJ10}. The reader can also  consult p. 187 in \cite{AmariCich2010} where our current use of Bregman divergence is discussed  under the name ``$\beta$-divergence''. It contains both Kullback-Leibler and Itakura-Saito divergences and
has found earlier applications in classical robustness and in machine learning.  Similar comparisons are also done  on p. 9 of \cite{DeyJ10}  which  demonstrates that the polynomial divergence of which our chosen Bregman divergence is based, is a strictly monotone increasing function of the R\'{e}nyi entropy. Hence minimizing the one or the other delivers essentially the same result.\\

 In this paper, we  prefer to work with the Bregman divergence in  its {\it functional form} (\ref{FBD}) when describing an ambiguity set around certain nominal probability distribution. Earlier approaches have described ambiguity sets using support and moment information or structural properties such as symmetry, unimodality, etc. We suggest that in the setting of index tracking, the uncertainty around the nominal model is difficult to specify accurately enough and for this reason considering a ``full'' functional-type neighbourhood is more suitable to start  with. It is remarkable that the objective function for the robust index tracking problem in this functional setting is  tractable and allows us to derive the optimal robust portfolio in a semi-analytical form as a solution of a system of nonlinear equations.

Since our purpose is to robustify the inference, it is {\it essential} for us to use the Bregman divergence (in particular, in its specific form of $\beta$- divergence). In  \cite{MiEguchi2002} the $\beta$-divergence  is used (which is essentially the same as the density power divergence of \cite{BHHJ98}) and it is explicitly demonstrated that using the global $\beta$-divergence leads to parameter estimators that are robust in the classical sense i.e., their influence function is bounded.\\

Now, by choosing the following weight function
\begin{eqnarray}
  v(x) = \frac{f(x)}{f(x)^{\lambda + 1}},                                 \nonumber
\end{eqnarray}

we have constructed the following functional Bregman divergence:
\begin{eqnarray}\label{eqn:Bregdiv}
  D_{Breg}(f,g)  & :=  &    \displaystyle\int_{\mathbb{R}^d} d_{F}(f(x),g(x)) v(x)dx \nonumber \\
                         &  =  &    \int \Big( \frac{g^{\lambda+1}(x)}{\lambda} - \frac{(\lambda+1)g(x)f^{\lambda}(x)}{\lambda}
                                  + f^{\lambda+1}(x)\Big)\frac{f(x)}{f(x)^{\lambda + 1}} dx   \nonumber \\
                         &  =  &    \int \Big(\frac{1}{\lambda}\mathcal{E}^{\lambda+1} - \frac{\lambda+1}{\lambda}\mathcal{E} + 1\Big)f(x) dx   \nonumber \\
                         &  =  &    \mathbb{E}\Big(\frac{1}{\lambda}\mathcal{E}^{\lambda+1} - \frac{\lambda+1}{\lambda}\mathcal{E} + 1\Big)   \nonumber \\
                         &  =  &    \mathbb{E}\big(G(\mathcal{E})\big),
\end{eqnarray}
where $G(\mathcal{E})=\frac{1}{\lambda}\mathcal{E}^{\lambda+1} - \frac{\lambda+1}{\lambda}\mathcal{E} + 1.$\\

From now on, we will always denote the nominal distribution's density  by $f$. Expectations in this paper are always supposed to be taken with respect to the nominal distribution $f$. Then there will be a one-to-one correspondence between the pair $(f(x), g(x))$ and the pair $(f(x), \mathcal{E}(x)) $ where
\begin{eqnarray}
  \mathcal{E}(x) = \frac{g(x)}{f(x)}.                        \nonumber
\end{eqnarray}
Hence by slight abuse of notation, we will shorten the notation by replacing
$
D_{Breg}(f,g)$ with $D_{Breg}(\mathcal{E}).$

%We emphasize that in  expression (\ref{eqn:Bregdiv}) some abuse of notation is %utilized. Indeed, it is not only the ratio  $\mathcal{E}$  that defines this %functional  divergence. The  distribution with respect to which the expectation is %taken, is also a part of the definition. Since in this paper the expectations are %always supposed to be taken with respect to the nominal  distribution $f$, we have %shortened the notation in this way.\\
\begin{remark}
We comment on our choice of the weight function $  v(x) = \frac{f(x)}{f(x)^{\lambda + 1}}$ in the definition of the functional Bregman divergence. Its usefulness goes much beyond the obvious fact that it leads to a simple expression for $D_{Breg}(\mathcal{E}).$
In  \cite{VLAN11} it is argued that re-scaling at each $x$ the usual Bregman divergence can bring about intrinsic robustification. The way the Bregman divergence is re-scaled in  \cite{VLAN11} leads to the so-called Total Bregman Divergence,
which is too computationally heavy to be implemented in our  case of functional Bregman divergence. However, the main idea of using a re-scaling that is inversely dependent on the derivative of the $F_{\lambda}(\cdot)$ function can be applied also in our case and suggests the above choice of $v(x).$ Our numerical experiments support this choice.

\end{remark}

\begin{remark}\label{remark:klconverge}
  One essential advantage of the special form of the Bregman divergence we use is that by just one parameter ($\lambda >0)$ we can control the deviation from the well-known Kullback-Leibler (KL) divergence $\mathbb{E}\Big(\mathcal{E}\log(\mathcal{E})\Big).$ The latter is obtained as a limiting case when  $\lambda \rightarrow 0. $  Indeed, as $\lambda \rightarrow 0$, we see that
  \begin{eqnarray}
     \frac{z^{1+\lambda} - (\lambda+1)z}{\lambda} & \rightarrow &  z\log(z)- z.         \nonumber
  \end{eqnarray}
  This limit result then yields
  \begin{eqnarray}
                      \mathbb{E}\Big(\frac{1}{\lambda}\mathcal{E}^{\lambda+1} - \frac{\lambda+1}{\lambda}\mathcal{E} + 1\Big)
     & \rightarrow &  \mathbb{E}\Big(\frac{g(x)}{f(x)}\log\Big(\frac{g(x)}{f(x)}\Big) - \frac{g(x)}{f(x)} + 1\Big)   \nonumber \\
     &      =      &  \mathbb{E}\Big(\frac{g(x)}{f(x)}\log\Big(\frac{g(x)}{f(x)}\Big)\Big)   \nonumber  \\
     &      =      &  \mathbb{E}\Big(\mathcal{E}\log(\mathcal{E})\Big).   \nonumber
  \end{eqnarray}
The value of $\lambda >0$ parameterizes a whole class of  Bregman divergences and dictates the extent to which the robust method differs from the non-robust method (in Basu et al. (1998) it is called an algorithmic parameter). By varying the value of $\lambda >0$ we can achieve a compromise between robustness and efficiency, as is  standard in robust statistic setting. Larger values of $\lambda$ correspond to a stronger emphasis on robustness. These would be useful when there is a belief that the divergence  between the nominal distribution and the actual distribution of the returns might be large.
\end{remark}

In a special case where both the nominal distribution and the  actual distribution are multivariate normal, the above Bregman divergence can be calculated in a closed form.\\

\noindent \textbf{Example}
[Multivariate Normal].
Suppose that the nominal distribution is a $d$-dimensional multivariate normal distribution $N_{d}(\bm{\mu}_{1},\bm{\Sigma}_{1})$ and an actual distribution is another $d$-dimensional multivariate normal
distribution $N_{d}(\bm{\mu}_{2},\bm{\Sigma}_{2}) . $ Then, the Bregman divergence as defined above can be calculated in closed form as:
\begin{eqnarray}
\label{fullform}
             D_{Breg}(\mathcal{E})
   &  =  &   \frac{1}{\lambda}\Big(\frac{\big(\det(\bm{\Sigma}_{1}\bm{\Sigma}_{2}^{-1})\big)^{\lambda+1}\det(\tilde{\bm{\Sigma}}_{\lambda})}{\det(\bm{\Sigma}_{1})}\Big)^{\frac{1}{2}}
            \exp\Big(-\frac{\lambda+1}{2}\bm{\mu}_{2}^{\intercal}\bm{\Sigma}_{2}^{-1}\bm{\mu}_{2}     \nonumber \\
   &     & +  \frac{\lambda}{2}\bm{\mu}_{1}^{\intercal}\bm{\Sigma}_{1}^{-1}\bm{\mu}_{1} +            \frac{1}{2}\tilde{\bm{\mu}}_{\lambda}^{\intercal}\tilde{\bm{\Sigma}}_{\lambda}^{-1}\tilde{\bm{\mu}}_{\lambda}\Big) - \frac{1}{\lambda},
\end{eqnarray}
provided $\tilde{\bm{\Sigma}}_{\lambda}$ is positive definite, where
\begin{eqnarray}
   \tilde{\bm{\Sigma}}_{\lambda}  &  =  &  \big((\lambda+1)\bm{\Sigma}_{2}^{-1} - \lambda\bm{\Sigma}_{1}^{-1}\big)^{-1},                                                               \nonumber \\
      \tilde{\bm{\mu}}_{\lambda}  &  =  &  \tilde{\bm{\Sigma}}_{\lambda}\big((\lambda+1)\bm{\Sigma}_{2}^{-1}\bm{\mu}_{2} - \lambda\bm{\Sigma}_{1}^{-1}\bm{\mu}_{1}\big).                \nonumber
\end{eqnarray}
If $\bm{\Sigma}_{1} = \bm{\Sigma}_{2} = \bm{\Sigma}$, the above formula simplifies to
\begin{eqnarray}\label{eqn:bregdivsimp}
             D_{Breg}(\mathcal{E})
   &  =  &   \frac{1}{\lambda}\exp\Big(\frac{\lambda(\lambda+1)}{2}(\bm{\mu}_{2} - \bm{\mu}_{1})^{\intercal}\bm{\Sigma}^{-1}(\bm{\mu}_{2} - \bm{\mu}_{1})\Big) - \frac{1}{\lambda}.
\end{eqnarray}
%\end{ex}

It is worth noting that if we set $\lambda \to 0$ in     (\ref{eqn:bregdivsimp}) we get
$$ \frac{1}{2} ({\mu}_{2} - \bm{\mu}_{1})^{\intercal}\bm{\Sigma}^{-1}(\bm{\mu}_{2} - \bm{\mu}_{1}),$$
i.e., one half of the squared Mahalanobis distance between the multivariate normal distributions   $N_{d}(\bm{\mu}_{1},\bm{\Sigma})$ and
 $N_{d}(\bm{\mu}_{2},\bm{\Sigma}).$                                      \\
%\end{ex}
\subsection{Robust Index Tracking} \label{mrit}
To perform index tracking in a robust way, we need to consider perturbations of the nominal distribution of the index. These perturbed distributions can be  contained inside a  ball of certain radius around the nominal distribution.
To this end, let us construct a Bregman divergence ball.  Suppose that $f$, the so-called nominal distribution, is the joint density of the $\ell$ assets in the index, and $g$ is the density of a perturbation of $f$,
we denote by $S_f$ the set of all functions $\mathcal{E} : \mathbb{R}^{d} \to \mathbb{R} $ representable in the form
$\mathcal{E} (x)=\frac{g(x)}{f(x)},$ where $g$ is a density. A Bregman divergence ball around of radius $\eta >0$ around $f$ is defined as
\begin{eqnarray}\label{eqn:klball}
  \mathcal{B} := \mathcal{B}(\eta) = \{g: \mathcal{E} \in S_f \hbox{ and } D_{Breg}(\mathcal{E}) \leq \eta\}.
\end{eqnarray}

We stress again that all moments throughout this paper are defined with respect to the nominal distribution. Then the robust version of the control problem
(\ref{eqn:nonrobust}) is defined as
\begin{eqnarray}\label{eqn:robust}
   V  &  =  &    \displaystyle\sup_{\bm{u} \in \mathcal{U}}\inf_{\bm{\mathcal{E}} \in \mathcal{B}} J(\mathcal{E},\bm{u}),
\end{eqnarray}
where
\begin{eqnarray}
           J(\mathcal{E},\bm{u})
  &  =  &  \mathbb{E}\Big(-\mathcal{E}\big(\bm{R}^{\intercal}\bm{u} - B\big)^{2}\Big), \nonumber
\end{eqnarray}
In the next section, we will derive a semi-analytical form of the optimal portfolio under the constructed Bregman divergence.

\subsection{Robust Optimal Portfolio under Bregman Divergence} \label{sec:semianalyopall}
The robust optimal index tracking portfolio can be obtained by applying the following result.
\begin{theorem}\label{thm:bregmansystem}
  For a fixed, small enough $\lambda > 0 ,$ if there exist $\alpha^{\ast} > 0$, $\beta^{\ast}$ and $\theta^{\ast}$ such that
  \begin{eqnarray}
                 \theta^{\ast}\bm{1}
    &   =   &    \mathbb{E}\Bigg(\frac{\partial{H}}{\partial \bm{u}}\Big(\frac{\lambda}{\lambda+1}\Big(\frac{-\beta^{\ast} - H(\bm{u})}{\alpha^{\ast}}\Big) + 1\Big)^{\frac{1}{\lambda}}\Bigg),        \nonumber \\
                 \bm{1}^{\intercal}\bm{u}
    &   =   &    1,                                              \nonumber    \\
                 \mathbb{E}\big(G(\mathcal{E}^{\ast})\big)
    &   =   &    \eta,                                       \nonumber    \\
                 \mathbb{E}\big(\mathcal{E}^{\ast}\big)
    &   =   &    1,   \nonumber
  \end{eqnarray}
  where
  \begin{eqnarray}
         H(\bm{u})   &    =    &  - \big(\bm{R}^\intercal\bm{u} - B\big)^{2},    \nonumber \\
    \mathcal{E}^{\ast}   &    =    &    \Big(\frac{\lambda}{\lambda+1}\Big(\frac{-\beta^{\ast} - H(\bm{u})}{\alpha^{\ast}}\Big) + 1\Big)^{\frac{1}{\lambda}},             \nonumber
  \end{eqnarray}
  then $\bm{u}$ is an optimal index tracking portfolio.
\end{theorem}

\begin{proof}
  Using the definition of  $H(\bm{u})$
%  \begin{eqnarray}
     %  \end{eqnarray}
  we see that (\ref{eqn:robust}) becomes
  \begin{eqnarray}\label{eqn:robuststep1simp}
     &  \displaystyle\sup_{\bm{u} \in \mathcal{U}}\inf_{\mathcal{E} \in \mathcal{B}} &  \mathbb{E}\big(\mathcal{E}H(\bm{u})\big).
  \end{eqnarray}
  To solve the inner optimization problem, we first write down the Lagrangian. For a fixed $\alpha \geq 0$ and a $\beta \in \mathbb{R}$, the Lagrangian is
  \begin{eqnarray}
    L^{inner}(\mathcal{E},\bm{u})  &  =  &  \mathbb{E}\Big(\mathcal{E}H(\bm{u}) + \alpha(G(\mathcal{E}) - \eta) + \beta(\mathcal{E} - 1)\Big),     \nonumber
  \end{eqnarray}
  where
  \begin{eqnarray}
    G(\mathcal{E})   &  =  &  \frac{1}{\lambda}\mathcal{E}^{\lambda+1} - \frac{\lambda+1}{\lambda}\mathcal{E} + 1.    \nonumber
  \end{eqnarray}
  Differentiating inside  the expectation and setting the result equal to zero yields:
  \begin{eqnarray}\label{eqn:LagrangianeqninsideexpN1}
    0  & = &  H(\bm{u})+ \alpha G^{'}(\mathcal{E}) + \beta.
  \end{eqnarray}

  Solving this equation, we obtain
  \begin{eqnarray}
  \label{semiclosed}
    \mathcal{E}^{\ast}  & = &  (G^{'})^{-1}\Big(\frac{-\beta - H\big(\bm{u}\big)}{\alpha}\Big),
  \end{eqnarray}
  provided that $\alpha \neq 0$. \\

  Next, we verify that this is indeed an optimal solution. The proof follows similarly to \cite[proposition 2.3]{GlaX14} and  \cite[theorem 2]{DeyJ10}. The idea is to show that along any feasible direction the value of the Lagrangian can not be optimized any further. \\

 Choose an arbitrary $\mathcal{E}.$ For $t \in [0,1]$, define $\hat{\mathcal{E}} = t\mathcal{E}  + (1-t)\mathcal{E}^{\ast}$, then we have
  \begin{eqnarray}
    L^{inner}(\hat{\mathcal{E}},\bm{u})  &  =  &   \mathbb{E}\Big(H(\bm{u})\hat{\mathcal{E}} + \alpha(G(\hat{\mathcal{E}}) - \eta) +  \beta(\hat{\mathcal{E}} - 1)\Big) \nonumber \\
                                   &  =  &   \mathbb{E}\Big(H(\bm{u})\big(t\mathcal{E}  + (1-t)\mathcal{E}^{\ast}\big) +  \alpha\big(G\big(t\mathcal{E} + (1-t)\mathcal{E}^{\ast}\big)
                                           - \eta\big)   \nonumber \\
                                   &     & + \beta\big(\big(t\mathcal{E}  + (1-t)\mathcal{E}^{\ast}\big) - 1\big)\Big). \nonumber
  \end{eqnarray}

  If we consider $L^{inner}(\hat{\mathcal{E}},\bm{u})$ as a function of $t$, and define
  \begin{eqnarray}
    K(t) &  =  &   \mathbb{E}\Big(H(\bm{u})\big(t\mathcal{E}  + (1-t)\mathcal{E}^{\ast}\big) +  \alpha\big(G\big(t\mathcal{E}  + (1-t)\mathcal{E}^{\ast}\big) - \eta\big)   \nonumber \\
         &     & + \beta\big(\big(t\mathcal{E}  + (1-t)\mathcal{E}^{\ast}\big) - 1\big)\Big), \nonumber
  \end{eqnarray}
  it is then easy to calculate

  \begin{eqnarray}
    K^{'}(t) &  =  &  \mathbb{E}\Big(\big(H(\bm{u}) + \alpha G^{'}\big(t\mathcal{E} + (1-t)\mathcal{E}^{\ast}\big)+ \beta\big)\big(\mathcal{E} - \mathcal{E}^{\ast}\big)\Big). \nonumber
  \end{eqnarray}
  This implies then
  \begin{eqnarray}
    K^{'}(0) &  =  & \mathbb{E}\Big(\big(H(\bm{u}) + \alpha G^{'}(\mathcal{E}^{\ast})  + \beta\big)\big(\mathcal{E} - \mathcal{E}^{\ast}\big)\Big) = 0, \nonumber
  \end{eqnarray}
  since $\mathcal{E}^{\ast}$ satisfies (\ref{eqn:LagrangianeqninsideexpN1}). In addition, we know that $L^{inner}$ is convex in its first argument, thus $K$ is convex in $t$ which implies
  $t = 0$ is an optimal solution. Because $\mathcal{E}$ is arbitrary, we can not improve the value of the objective along any feasible direction from $\mathcal{E}^{\ast}$. This concludes
  that $\mathcal{E}^{\ast}$ is an optimal solution. \\

  Next, we notice that the set
  \begin{eqnarray}
    \Big\{\mathcal{E}: \mathbb{E}\big(G(\mathcal{E})\big) < \eta\Big\} \nonumber
  \end{eqnarray}
  is not empty. By Theorem 2.1. in \cite{BenTC88}, strong duality holds. This implies (see for example, pp. 242--243 in \cite{BoyV09}) that the optimal solution $\mathcal{E}^{\ast}$ and its
  corresponding $\alpha$ satisfies the following system:
  \begin{eqnarray}
    \alpha\mathbb{E}\big(G(\mathcal{E}^{\ast}) - \eta\big)  &    =   &  0, \nonumber \\
                 \mathbb{E}\big(G(\mathcal{E}^{\ast})\big)  &  \leq  &  \eta, \nonumber \\
                                                   \alpha   &    >   &  0.    \nonumber
  \end{eqnarray}
  We denote the solution $\alpha$ of this system as $\alpha^{\ast}$. Thus, we obtain
  \begin{eqnarray}\label{eqn:optimdiv}
    \mathcal{E}^{\ast}  &  =  &   \Big(\frac{\lambda}{\lambda+1}\Big(\frac{-\beta^{\ast} - H(\bm{u})}{\alpha^{\ast}}\Big) + 1\Big)^{\frac{1}{\lambda}},
  \end{eqnarray}
  where $\beta^{\ast}$ is the solution of $\mathbb{E}(\mathcal{E}^{\ast}) = 1$. With an appropriate choice of (small) $\lambda$, we can always achieve that $\mathcal{E}^{\ast}$ in (\ref{eqn:optimdiv})  is well-defined and positive. Indeed,
    %  if we have
  %\begin{eqnarray}\label{eqn:divpositivity}
  %  \frac{\lambda}{\lambda+1}\Big(\frac{-\beta^{\ast} - H(\bm{u})}{\alpha^{\ast}}\Big) + 1 > 0,
  %\end{eqnarray}
  %then $\mathcal{E}^{\ast} > 0$. The requirement that (\ref{eqn:divpositivity}) holds is equivalent to require
  %\begin{eqnarray}
  %  \frac{\beta^{\ast}}{\alpha^{\ast}}   &  <  &   1 + \frac{1}{\lambda}.                  \nonumber
  %\end{eqnarray}
  %The latter inequality can be guaranteed with a carefully chosen, small enough $\lambda .$     \\
for the limiting case $\lambda=0$ (corresponding to the KL divergence) we have
 $$G(\mathcal{E})=\mathcal{E}\log \mathcal{E}-\mathcal{E} +1, G'(\mathcal{E})=\log(\mathcal{E})$$
  and the solution $\mathcal{E}^{\ast}=e^{\frac{-\beta^{\ast}-H(\bm{u})}{\alpha^{\ast}}}$ clearly being positive.
  As in Proposition 2.3 of  \cite{GlaX14} we can argue that by continuity we can find a set $[0,\lambda_0]$ such that $\mathcal{E}^{\ast}    =     \Big(\frac{\lambda}{\lambda+1}\Big(\frac{-\beta^{\ast} - H(\bm{u})}{\alpha^{\ast}}\Big) + 1\Big)^{\frac{1}{\lambda}}$ will stay positive for any $\lambda \in [0,\lambda_0].$\\

  Using (\ref{eqn:optimdiv}), we end up with the optimization problem
  \begin{equation}
  \label{outer}
\sup_{\bm{u} \in \mathcal{U}} J(\mathcal{E}^{\ast},\bm{u}).
 \end{equation}
 In Appendix \ref{finishing} we demonstrate that the solution to this optimization problem satisfies the system of equations  in the statement of Theorem \ref{thm:bregmansystem}.

\end{proof}

As mentioned in \hyperref[remark:klconverge]{Remark \ref{remark:klconverge}}, as $\lambda \rightarrow 0$, the chosen Bregman divergence converges to the Kullback-Leibler (KL) divergence. Indeed,
as $\lambda \rightarrow 0$, the system in Theorem \ref{thm:bregmansystem} also converges to the corresponding system of the KL divergence. This is summarized in the
following result.

\begin{remark}
We have assumed existence of $\alpha^* >0$ in Theorem \ref{thm:bregmansystem} and this  is exploited in the resulting presentation of $\mathcal{E}^*.$ It is clear, however, that the case $\alpha=0$ should  safely be excluded. Indeed if $\alpha $ was zero then the restriction about
$g$ belonging to the ball of radius $\eta$ around $f$ is ignored. Then
$L^{inner}(\mathcal{E}, \bm{u})$ implies to minimize  $\mathbb{E}(\mathcal{E}H(\bm{u})),$ where $H(\bm{u})$  is a negative random variable, under the \emph{only} restriction that $\mathbb{E}(\mathcal{E})=1 .$  This is equivalent to ask to minimize
 $\mathbb{E}_{g}H(\bm{u})$ where $\mathbb{E}_{g}$ stands for calculating the expected value under the ``arbitrary''
  actual distribution. This problem does not have a solution since for any specified $g^*$ such that $\mathbb{E}_{g^*}H(\bm{u})=A$ we can find another $\tilde{g}$ such that $\mathbb{E}_{\tilde{g}}H(\bm{u})<A$ as long as $\tilde{g}$ puts higher mass at the negative values of $H(\bm{u})$ with a large magnitude.\\

 It is also important to note that for the chosen $\lambda >0$ in Theorem \ref{thm:bregmansystem} the solution of the equation system for $\alpha^*>0, \beta^*$ and $\theta^*$ is also required to satisfy the condition $\frac{\beta^*}{\alpha^*}<1+\frac{1}{\lambda}. $ Finding precise conditions which  $\lambda >0$ must satisfy
  under a general nominal distribution seems to be very difficult. However, we can state that numerically, in all examples that we have tried, we have  observed that if certain $\lambda '$ is found which ``works'', then   all values $\lambda \in (0,\lambda ')$ also work, i.e.,  for them, respective $\alpha^*>0, \beta^*$ and $\theta^*$ satisfying the system of equations exist  with the inequality $$\frac{\beta^*}{\alpha^*}<1+\frac{1}{\lambda}$$ being satisfied.

\end{remark}

\begin{corollary}
\label{corolabel}
  Suppose that there exist $\alpha^{\ast} > 0$, $\beta^{\ast}$ and $\theta^{\ast}$ such that
  \begin{eqnarray}
                 \bm{1}^{\intercal}\theta^{\ast}
    &   =   &    \mathbb{E}\Bigg(\frac{\partial{H}}{\partial \bm{u}}\exp\Big(\frac{-\beta^{\ast} - H(\bm{u})}{\alpha^{\ast}}\Big)\Bigg),  \nonumber \\
                 \bm{1}^{\intercal}\bm{u}
    &   =   &    1,                                              \nonumber    \\
                 \mathbb{E}\big(G(\mathcal{E}^{\ast})\big)
    &   =   &    \eta,                                           \nonumber    \\
                 \mathbb{E}\big(\mathcal{E}^{\ast}\big)
    &   =   &    1,   \nonumber
  \end{eqnarray}
  where
  \begin{eqnarray}
            H(\bm{u})   &  =  &  - \big(\bm{R}^\intercal\bm{u} - B\big)^{2},  \ \ \textnormal{and} \ \
    \mathcal{E}^{\ast}     =       \exp\Big(\frac{-\beta^{\ast} - H(\bm{u})}{\alpha^{\ast}}\Big),  \nonumber
  \end{eqnarray}
  then $\bm{u}$ is an optimal index tracking portfolio.
\end{corollary}
\noindent The proof of the Corollary is obtained by taking a limit as $\lambda \to 0$ in Theorem \ref{thm:bregmansystem}.
\begin{remark}
The main part of the numerical procedure of our method is the implementation of the solution of the nonlinear equation system from Theorem \ref{thm:bregmansystem}. We used MATLAB and applied the interior-point method for solving box-constrained nonlinear systems. Initial guesses for the solution were necessary to be chosen for the portfolio's weights, and  for the parameters $\alpha, \beta$ and $\theta. $ Uniform initial weights turned out to be working fine all the time. The initial weights for the remaining parameters required more careful choosing and experimentation. The ones that worked fine for our numerical experiments were $\alpha =0.02, \beta=0.01$ and $\theta =-0.05.$
\end{remark}

\section{Modified myopic robust index tracking}
\label{emrit}

The discussion in Section \ref{sec:semianalyopall} can be extended to cover a specific form of modified myopic robust index tracking problem. In the previous section we measured the quality of the index tracking by using the quadratic loss function since this is the typical choice in the portfolio tracking literature. This choice equally penalizes  the performance of the portfolio whenever it deviates by the same magnitude irrespectively of  whether the deviation is above or below the value of the index.\\

The main focus in \cite{dePdC16} is on formulating an optimization problem that represents  a balancing of the trade-off between tracking error and excess return.
A different goal may be of interest in a robust setting. Typically, in the latter setting, the goal is to safeguard against worst-case scenarios and the solution obtained reflects this goal. Hence it is expected to give superior performance, especially in a downturn market. If for various reasons the investor still remains in the market during a downturn (for example, expecting that this downturn would be relatively short-lived, or because of limited liquidity),  one would not be willing to penalize if the portfolio outperforms the index in such cases.
 Obviously, a more reasonable choice to replace the loss $\ell (x)=x^2$ to be used in such a situation would be  based, for example, on a smooth approximation of the function $\ell_1(x)= x^2 \hbox{ if } x>0 \hbox { and 0 else}.$
Other choices also make sense, for example $\ell_2(x)=[x]_{+} .$ Direct utilization of these types of functions makes a lot of sense since we do not really want to penalize  when the portfolio happens to outperform the index.\\

 However, there is a technical difficulty to overcome if we want to include such type of losses in our approach. It is related to the fact that the functions
$\ell_1$ and $\ell_2$ are not smooth at the origin.
If we would like to utilize the steps as in Theorem \ref{thm:bregmansystem} and show that the Hessian is negative semi-definite,  we need a convex twice differentiable loss function $\ell_i(x), i=1,2$ to replace $\ell(x). $
Also, from a technical prospective, the gradient of $H$  should be possible to calculate, preferably in a closed form.
 We suggest the function $\tilde{\ell}_1(x)=\frac{1}{\epsilon}\int_0^{\infty}\phi(\frac{1}{\epsilon}(x-t))\ell_1(t)dt$
 with a suitably chosen small $\epsilon >0$  as  approximation for $\ell_1 (x).$ For approximation of $\ell_2(x),$ the expression $\tilde{\ell}_2(x)=x+\epsilon \log (1+e^{-x/\epsilon})$ from the literature (see e.g., \cite{CM95}) can be used and
 is known as  ``the neural networks smooth plus function''. Using these, the function
  $H(\bm{u})= - \big(\bm{R}^\intercal\bm{u} - B\big)^{2}=- \big(B-\bm{R}^\intercal\bm{u} \big)^{2}$ in Theorem \ref{thm:bregmansystem} can be replaced by
  $\tilde{H}_1(\bm{u})= - \tilde{\ell}_1(B-\bm{R}^\intercal\bm{u})$ or by $\tilde{H}_2(\bm{u})= - \tilde{\ell}_2(B-\bm{R}^\intercal\bm{u} ),$ respectively. The corresponding gradient of $H$ is to be replaced by the gradient of  $\tilde{H}_1$ or $\tilde{H}_2$ and these
  are easily calculated by using the chain rule and the derivatives of one-dimensional argument for $\tilde{\ell_1}(t)$ and $\tilde{\ell_2}(t).$
   Both derivatives of $\tilde{H}_1$ or $\tilde{H}_2$ w.r.t. the components of $\bm{u}$  deliver smooth approximating functions.  We prefer the first approximation since its second mixed derivatives appear to be varying more smoothly around the origin for small values of $\epsilon .$
Elementary calculation of the integral gives the following approximations for the function $\tilde{\ell}_1(x)$ :
\begin{equation}
\label{smoothell1}
\tilde{\ell}_1(x)=x^2\Phi(\frac{x}{\epsilon})+2x\epsilon\phi(\frac{x}{\epsilon})+\epsilon^2[\frac{1}{2}+\frac{1}{2}Erf(\frac{x}{\sqrt{2}\epsilon})-
\frac{x}{\epsilon}\phi(\frac{x}{\epsilon})]
\end{equation}
Here $\Phi(\cdot )$ denotes the cumulative distribution function of the univariate standard normal distribution, $\phi(\cdot )$ denotes the density and the $Erf(\cdot )$ function is defined as $Erf(x)=\frac{2}{\sqrt{\pi}}\int_0^x e^{-t^2}dt.$ Having in mind  the relationship $\frac{1}{2}+\frac{1}{2}Erf(\frac{x}{\sqrt{2}\epsilon})=\Phi(\frac{x}{\epsilon}),$ (\ref{smoothell1}) simplifies further to the explicit expression
\begin{equation}
\label{smoothell1simple}
\tilde{\ell}_1(x)=(x^2+\epsilon^2)\Phi(\frac{x}{\epsilon})+x\epsilon\phi(\frac{x}{\epsilon}).
\end{equation}

Differentiating (\ref{smoothell1})  delivers the resulting approximations for the derivatives of  $\tilde{\ell}_1(x)$ :
\begin{equation}
\tilde{\ell}_1'(x)=2x\Phi(\frac{x}{\epsilon})+2\epsilon \phi(\frac{x}{\epsilon}), \hbox{    } \tilde{\ell}_1''(x)=2\Phi(\frac{x}{\epsilon}).
\end{equation}
Of course, the approximations for the derivatives of $\tilde{\ell}_2(x)$ are:
$$
\tilde{\ell}_2'(x)=\frac{1}{1+\exp(-x/\epsilon)} , \hbox{   } \tilde{\ell}_2''(x)=\frac{1}{\epsilon}\frac{e^{x/\epsilon}}{(1+e^{x/\epsilon})^2}.
$$

\section{Numerical Analysis}
\label{numerics}
In this section, we perform various numerical comparisons to illustrate the usefulness of our model.
It is worth noting that there is  some difference between the general theory (in particular, the statements in Section \ref{myopic})  and the way to illustrate the theory via examples. We stress that, as seen in Section \ref{myopic}, the main theoretical statement in Theorem \ref{thm:bregmansystem} on the basis of which  our numerical procedure is  implemented, does not explicitly require calculation of the least favourable distribution  in the Bregman ball; all that is needed is the radius $\eta$ of the ball.  If we wanted to  illustrate the full effect of the robustification theory on a particular example, we should ideally be able to calculate the least favourable distribution and simulate from it. Determining the least favourable distribution is  very difficult  even if the nominal distribution was multivariate normal. On the other hand, we know that the least favourable distribution is on the surface of the ball (since the Lagrange multiplier $\alpha $ is not equal to zero). Hence,  just for the purpose of generating illustrative examples, we have chosen a distribution that has the  maximal allowed divergence from the nominal {\it and}  is possible to deal with (e.g., in Example 5.1, we choose it to be multivariate normal with the same covariance matrix as the nominal but with a re-scaled mean). Of course, this distribution is not necessarily the least favourable  though it is on the maximal allowable distance from the nominal distribution.  This approach allows us to simulate our toy examples.  We follow  traditional approaches in the robustness literature where a contaminated neighbourhood is typically a neighbourhood of the multivariate normal distribution, e.g., \cite{HubRon2009}. For this reason, we choose our testing of the statements of Theorem \ref{thm:bregmansystem} to involve  perturbations of multivariate normal and of multivariate {\tt{t}}. We are aware of the fact that these simulations represent a simplistic indicative study of the effect of the general statement of  Theorem \ref{thm:bregmansystem}. A more comprehensive numerical work would be required to test our procedure by examining Bregman neighbourhoods of other nominal distributions of interest in finance. As this paper is more on the methodological side, such numerical work is beyond the scope of the presented research. We hope that the reader will still be able to get some insight about the usefulness of our approach by examining the presented numerical examples.

One more point to make is that the least favourable distribution corresponds to a pessimistic scenario that may not be the one that would actually happen hence our examples, by not directly simulating from it, hopefully still  give useful insight in the performance of the robust procedure.

\subsection{Performance Comparison via simulation: Index Tracking}
 We first compare the performance of the robust  and the non-robust portfolio in the context of index tracking.
Suppose that we have an index which is made up of five assets according to the following weight vector:

\begin{eqnarray}\label{eqn:empmeanandcov}
   \bm{w} = \left(\begin{array}{c}  0.15  \\  0.20  \\ 0.20  \\ 0.15 \\ 0.30 \end{array}\right).   \nonumber
\end{eqnarray}
The expected return and the covariance matrix of these five assets are given by
\begin{eqnarray}\label{eqn:empmeanandcov}
   &    & \tilde{\bm{\mu}} = \left(\begin{array}{c}     0.0025  \\  0.0035  \\ 0.0010  \\   0.0005 \\  0.0045 \end{array}\right), \ \
          \tilde{\bm{\Sigma}} = \left(\begin{array}{ccccc}  0.0020  &  0.0000  &  0.0000  &  0.0000  &  0.0000   \\
                                                                 0.0000  &  0.0025  &  0.0000  &  0.0000  &  0.0000    \\
                                                                 0.0000  &  0.0000  &  0.0012  &  0.0000  &  0.0000     \\
                                                                 0.0000  &  0.0000  &  0.0000  &  0.0001  &  0.0000     \\
                                                                 0.0000  &  0.0000  &  0.0000  &  0.0000  &  0.0033
                                  \end{array}\right).
\end{eqnarray}
We will use the first four assets to track this hypothetical index. The following example is used to demonstrate  the comparison.\\

\noindent \textbf{Example} [Multivariate Normal (MVN)].
 Suppose that the nominal distribution is a $\ell$-dimensional multivariate normal distribution $N_{\ell}(\bm{\mu}_{1},\bm{\Sigma}_{1})$ and an actual distribution is an  $\ell$-dimensional multivariate normal distribution $N_{\ell}(\bm{\mu}_{2},\bm{\Sigma}_{2})$ which is on a Bregman divergence $\eta$ from the nominal. Fix $\lambda = 0.1$, we  assume $\bm{\Sigma}_{2} = \bm{\Sigma}_{1}$ and take $\eta = 0.1, 0.2, 0.5, 0.8, 1.0, 2.0, 5.0$, which results in $\bm{\mu}_{2} = k\bm{\mu}_{1}$ for some $k(\eta) \in \mathbb{R}$.\\
%\end{ex}

We  simulate $5,000,000$ returns from the nominal distribution to calculate the robust portfolio. The same number of simulations is used to draw samples from the actual distribution to make a comparison between the robust and non-robust portfolios. Several measures are used to accomplish the comparison. The first measure is the number of times (in percentage) that the robust case outperforms the non-robust case. We call this measure the Beating Time (BT). The larger the BT, the more times the robust portfolio  outperforms the non-robust one. The out-performance is in the sense of a lower tracking error, where the tracking error (TE) is defined as
\begin{eqnarray}
  TE  &  =  &   \big(\bm{u}^{\intercal}\bm{R} - B\big)^2 = \big(\bm{u}^{\intercal}(\bm{1}+\bm{r}) - B\big)^2.  \nonumber
\end{eqnarray}
Here $\bm{u}$ is a control (either robust or non-robust) applied to the tracking portfolio, $\bm{r}$ denotes the return of the assets in the tracking portfolio, and $B$ denotes the return of the index.
The second measure we apply to both controls is the expected tracking error (ETE). Obviously, the smaller the expected tracking error, the better the portfolio  (on average). We also  compare
the performance of the robust and non-robust portfolios with the index, and introduce a measure called the expected excess of index (EEI), where the excess of index (EI) is defined as
\begin{eqnarray}
    EI  &  =  &   \bm{u}^{\intercal}(\bm{1}+\bm{r}) - B.  \nonumber
\end{eqnarray}
Thus, a negative EI indicates the portfolio is beaten by the index. The EEI is the average of EI over the number of simulations performed. \\

\begin{table}[h]
\centering
%\resizebox{1.0\textwidth}{!}{\begin{minipage}{\textwidth}
\caption{Tracking performance: robust optimal solution (Bregman) versus non-robust optimal solution, $\lambda=0.1, k<0 $}\label{table:t1a1}
%\tiny
\begin{tabular}{cccccccc}
\hline
 \multicolumn{1}{c}{$\eta$}  &      BT      &            \multicolumn{3}{c}{ETE}          &             \multicolumn{3}{c}{EEI}        \\
\cline{3-5}\cline{6-8}
                             &              &   robust   &   non-robust &    difference   &   robust   &    non-robust   &  difference  \\
                             &     $\%$     &      \multicolumn{3}{c}{($*10^{-4}$)}       &       \multicolumn{3}{c}{($*10^{-4}$)}      \\
\cline{1-8}
    0.1  ($k =   -2.2158$)   &     50.51    &    3.1055  &     3.1055   &      0.0000     &   24.8692  &     24.8728    &   -0.0036   \\
    0.2  ($k =   -3.5366$)   &     51.67    &    3.2016  &     3.2017   &     -0.0001     &   39.7446  &     39.7556    &   -0.0110    \\
    0.5  ($k =   -6.1208$)   &     53.99    &    3.5174  &     3.5180   &     -0.0006     &   68.8326  &     68.8753    &   -0.0427    \\
    0.8  ($k =   -7.9434$)   &     55.25    &    3.8417  &     3.8431   &     -0.0014     &   89.3307  &     89.4115    &   -0.0808    \\
    1.0  ($k =   -8.9526$)   &     55.91    &    4.0573  &     4.0594   &     -0.0021     &  100.6753  &    100.7832    &   -0.1079    \\
    2.0  ($k =  -12.7653$)   &     58.25    &    5.1028  &     5.1099   &     -0.0071     &  143.4931  &    143.7447    &   -0.2516    \\
    5.0  ($k =  -19.5278$)   &     61.51    &    7.8510  &     7.8812   &     -0.0302     &  219.2462  &    219.9450    &   -0.6988   \\
\hline
\end{tabular}
%\end{minipage}}
\end{table}

\begin{table}[h]
\centering
%\resizebox{1.0\textwidth}{!}{\begin{minipage}{\textwidth}
\caption{Tracking performance: robust optimal solution (Bregman) versus non-robust optimal solution, $\lambda=0.1, k>0$}\label{table:t1a2}
%\tiny
\begin{tabular}{cccccccc}
\hline
 \multicolumn{1}{c}{$\eta$}  &      BT      &            \multicolumn{3}{c}{ETE}          &             \multicolumn{3}{c}{EEI}        \\
\cline{3-5}\cline{6-8}
                             &              &   robust   &   non-robust &    difference   &   robust   &    non-robust   &  difference  \\
                             &     $\%$     &      \multicolumn{3}{c}{($*10^{-4}$)}       &       \multicolumn{3}{c}{($*10^{-4}$)}      \\
\cline{1-8}
    0.1  ($k =    4.2158$)   &     52.10    &    3.2702  &     3.2702   &      0.0000     &  -47.5911  &     -47.5978    &    0.0067   \\
    0.2  ($k =    5.5366$)   &     54.20    &    3.4338  &     3.4340   &     -0.0002     &  -62.4633  &     -62.4805    &    0.0172    \\
    0.5  ($k =    8.1208$)   &     56.89    &    3.8817  &     3.8827   &     -0.0010     &  -91.5435  &     -91.6003    &    0.0568    \\
    0.8  ($k =    9.9434$)   &     58.00    &    4.2989  &     4.3011   &     -0.0022     & -112.0351  &    -112.1365    &    0.1014    \\
    1.0  ($k =   10.9526$)   &     58.56    &    4.5659  &     4.5691   &     -0.0032     & -123.3760  &    -123.5082    &    0.1322    \\
    2.0  ($k =   14.7653$)   &     60.64    &    5.8053  &     5.8149   &     -0.0096     & -166.1784  &    -166.4697    &    0.2913    \\
    5.0  ($k =   21.5278$)   &     63.36    &    8.8956  &     8.9325   &     -0.0369     & -242.6700  &    -241.8988    &    0.7712    \\
\hline
\end{tabular}
%\end{minipage}}
\end{table}

In this example, both the nominal and the actual distribution are MVN. From (\ref{eqn:bregdivsimp}), it is easy to see that the Bregman divergence between these two distributions is given by
\begin{eqnarray}
             \eta
   &  =  &   \frac{1}{\lambda}\exp\Big(\frac{\lambda(\lambda+1)}{2}(1-k)^2\bm{\mu}_{1}^{\intercal}\bm{\Sigma}^{-1}_{1}\bm{\mu}_{1}\Big)
           - \frac{1}{\lambda}.   \nonumber
\end{eqnarray}
This then implies
\begin{eqnarray}
  k  &  =  &   1 \pm \sqrt{\frac{\log(\eta\lambda+1)}{\frac{\lambda(\lambda+1)}{2}\bm{\mu}_{1}^{\intercal}\bm{\Sigma}^{-1}_{1}\bm{\mu}_{1}}},   \nonumber
\end{eqnarray}
which allows us to get the relevant $k$. \\

 If we take
\begin{eqnarray}
  k  &  =  &   1 - \sqrt{\frac{\log(\eta\lambda+1)}{\frac{\lambda(\lambda+1)}{2}\bm{\mu}_{1}^{\intercal}\bm{\Sigma}^{-1}_{1}\bm{\mu}_{1}}},   \nonumber
\end{eqnarray}
we obtain
%\hyperref[table:t1a1]{Table \ref{table:t1a1}}
Table \ref{table:t1a1}, and in the other case,
% \begin{eqnarray}
%   k  &  =  &   1 + \sqrt{\frac{\log(\eta\lambda+1)}{\frac{\lambda(\lambda+1)}{2}\bm{\mu}_{1}^{\intercal}\bm{\Sigma}^{-1}_{1}\bm{\mu}_{1}}},   \nonumber
% \end{eqnarray}
we obtain
%\hyperref[table:t1a2]{Table \ref{table:t1a2}}
Table \ref{table:t1a2}. \\

It can be seen that  in both Table \ref{table:t1a1} and Table \ref{table:t1a2}, the robust portfolio
outperforms the non-robust one when BT or ETE is used as a comparison measure. In contrast, when EEI is used, if there is a loss made, i.e., the portfolio underperforms the index, the robust portfolio safeguards and performs better. This leads to a positive difference  in
the last column of Table \ref{table:t1a2}. When there is a profit made, the opposite happens and a negative difference
is recognized as shown in  Table \ref{table:t1a1}.\\

Recall that the parameter $\lambda$ controls the amount of robustness applied: the smaller the $\lambda$, the less robustness effect. %Thus, we expect to see a less 'benefit' we obtain from robustification if we %decrease $\lambda$.
This belief is confirmed
from the results obtained in Table \ref{table:t1a3} and Table \ref{table:t1a4} when $\lambda$ is
taken to be 0.05. Attention should be directed at comparing the pairs: Table \ref{table:t1a3} with Table \ref{table:t1a1}, and Table \ref{table:t1a4} with Table \ref{table:t1a2}, respectively. It becomes apparent that, when the ball radius $\eta$ is small (hence no need of significant robustification), the performance is about the same no matter whether $\lambda =0.05 $ or $\lambda=0.1$ was used.
However, when $\eta$ is increased to, say, 2 or 5, more robustification is required and using the higher value of $\lambda=0.1$ proves to bring a higher percentage of BT.

\begin{table}[h]
\centering
%\resizebox{1.0\textwidth}{!}{\begin{minipage}{\textwidth}
\caption{Tracking performance: robust optimal solution (Bregman) versus non-robust optimal solution, $\lambda=0.05, k<0$}\label{table:t1a3}
%\tiny
\begin{tabular}{cccccccc}
\hline
 \multicolumn{1}{c}{$\eta$}  &      BT      &            \multicolumn{3}{c}{ETE}          &             \multicolumn{3}{c}{EEI}        \\
\cline{3-5}\cline{6-8}
                             &              &   robust   &   non-robust &    difference   &   robust   &    non-robust   &  difference  \\
                             &     $\%$     &      \multicolumn{3}{c}{($*10^{-4}$)}       &       \multicolumn{3}{c}{($*10^{-4}$)}      \\
\cline{1-8}
    0.1  ($k =   -2.2955$)   &     50.63    &    3.1101  &     3.1101   &      0.0000     &   25.7668  &     25.7716     &   -0.0048   \\
    0.2  ($k =   -3.6548$)   &     51.36    &    3.2124  &     3.2125   &     -0.0001     &   41.0720  &     41.0879     &   -0.0159    \\
    0.5  ($k =   -6.3327$)   &     52.38    &    3.5506  &     3.5515   &     -0.0009     &   71.1951  &     71.2628     &   -0.0677    \\
    0.8  ($k =   -8.2414$)   &     53.09    &    3.9019  &     3.9043   &     -0.0024     &   92.6362  &     92.7703     &   -0.1341    \\
    1.0  ($k =   -9.3074$)   &     53.53    &    4.1379  &     4.1416   &     -0.0037     &  104.5987  &    104.7813     &   -0.1826    \\
    2.0  ($k =  -13.4063$)   &     55.27    &    5.3099  &     5.3228   &     -0.0129     &  150.5243  &    150.9677     &   -0.4434    \\
    5.0  ($k =  -21.0432$)   &     58.17    &    8.6066  &     8.6615   &     -0.0549     &  235.8081  &    237.0205     &   -1.2124   \\
\hline
\end{tabular}
%\end{minipage}}
\end{table}

\begin{table}[h]
\centering
%\resizebox{1.0\textwidth}{!}{\begin{minipage}{\textwidth}
\caption{Tracking performance: robust optimal solution (Bregman) versus non-robust optimal solution, $\lambda=0.05, k>0$}\label{table:t1a4}
%\tiny
\begin{tabular}{cccccccc}
\hline
 \multicolumn{1}{c}{$\eta$}  &      BT      &            \multicolumn{3}{c}{ETE}          &             \multicolumn{3}{c}{EEI}        \\
\cline{3-5}\cline{6-8}
                             &              &   robust   &   non-robust &    difference   &   robust   &    non-robust   &  difference  \\
                             &     $\%$     &      \multicolumn{3}{c}{($*10^{-4}$)}       &       \multicolumn{3}{c}{($*10^{-4}$)}      \\
\cline{1-8}
    0.1  ($k =    4.2955$)   &     52.52    &    3.2788  &     3.2789   &     -0.0001     &  -48.4876  &    -48.4966    &    0.0090   \\
    0.2  ($k =    5.6548$)   &     53.38    &    3.4506  &     3.4509   &     -0.0003     &  -63.7882  &    -63.8128    &    0.0246    \\
    0.5  ($k =    8.3327$)   &     54.05    &    3.9254  &     3.9270   &     -0.0016     &  -93.8985  &    -93.9878    &    0.0893    \\
    0.8  ($k =   10.2414$)   &     54.70    &    4.3738  &     4.3776   &     -0.0038     & -115.3283  &   -115.4953    &    0.1670    \\
    1.0  ($k =   11.3074$)   &     55.10    &    4.6639  &     4.6695   &     -0.0056     & -127.2841  &   -127.5062    &    0.2221    \\
    2.0  ($k =   15.4063$)   &     56.78    &    6.0434  &     6.0606   &     -0.0172     & -173.1823  &   -173.6927    &    0.5104    \\
    5.0  ($k =   23.0432$)   &     59.40    &    9.7240  &     9.7904   &     -0.0664     & -258.4161  &   -259.7455    &    1.3294   \\
\hline
\end{tabular}
%\end{minipage}}
\end{table}

\subsection{Performance Comparison via simulation during market downturn}

 In this section, we illustrate the effect of using the loss function $\tilde{\ell}_1(.) $ from Section \ref{emrit} on two examples: the first  example involves the multivariate normal as a nominal distribution and the second example deals with multivariate $t$ as a nominal distribution.\\

\noindent \textbf{Example}  [Nominal multivariate normal].
%\end{ex}
 Given that all components of the chosen mean vector of the multivariate normal are positive, the single scalar multiplication with a value of $k<1$ represents  a market downturn scenario.   As $k$ is related to the radius $\eta ,$ a larger value of $\eta $ pushes $k$ further in the negative territory. The number of simulations used to calculate the robust portfolio and to assess the performance was kept at $1,000,000$ as a sufficient stabilization of the results was already appearing at this number of simulations.
We applied the smoothed loss function  $\tilde{\ell}_1(.) $ from Section \ref{emrit} with $\epsilon = 0.01$ and $\lambda =0.1.$ We varied the radius $\eta$ through the range $0.1, 0.2, 0.5, 0.8, 1, 2, 5$ as before and registered the percentage of cases in which the Bregman-based portfolio outperformed the non-robust one. (The non-robust portfolio was defined as minimizing the same loss $\tilde{\ell}_1(x)$ but without considering a neighbourhood around the nominal distribution).\\

 As expected, the percentage of cases in which the robust portfolio was not worse than the non-robust one was quite large.    The results are presented in Table \ref{table:t5}.\\

Similar results to the ones presented in Table \ref{table:t5} can be  obtained when $\tilde{\ell}_2(.)$ was used as a smoothed loss function.
The results clearly outline the significant benefits of using the robust portfolio in a market downturn scenario. Of course, this is to be expected by the nature of the optimization problem that is solved in the robust setting. We note that in  Table \ref{table:t5} the cases where both  portfolios deliver a zero value for the loss $\tilde{\ell}_1(\cdot)$ have been counted towards the percentage BT (since these cases are considered as ``not worse'' for the robust portfolio). One may suggest that it might be  fairer to  exclude theses cases from the comparison (i.e., to consider in what proportion of cases the robust portfolio delivered a truly better outcome). It would be expected that this proportion would be smaller but still high enough. Indeed this expectation is confirmed by the results that are presented in Table \ref{table:t6}.
\begin{center}
\begin{table}[h]
\centering
%\resizebox{1.0\textwidth}{!}{\begin{minipage}{\textwidth}
\caption{Tracking performance using the loss $\tilde{\ell}_1$: robust optimal solution versus non-robust optimal solution (MVN) $\lambda = 0.1, \epsilon=0.01$ (cases of both losses being zero included).}\label{table:t5}
%\tiny
\begin{center}
\begin{tabular}{cccccccc}
\hline
 \multicolumn{1}{c}{$\eta$}  &      BT               &         \multicolumn{3}{c}{ETE}         \\
\cline{3-5}
%\cline{6-8}
                             &                &   robust  &   non-robust      &  difference  \\
                             &     $\%$       &       \multicolumn{3}{c}{($*10^{-4}$)}  \\
\cline{1-5}
    0.1  ($k =    -2.2158$)  &     81.24      &   58.5587  &    58.9210  &    -0.3623    \\
    0.2  ($k =    -3.5366$)  &     84.06      &   52.3295  &    52.9541  &    -0.62464    \\
    0.5  ($k =    -6.1208$)  &     88.65      &   41.2324  &    42.3339  &    -1.1015    \\
    0.8  ($k =    -7.9434$)  &     91.46      &   34.5631  &    36.0577  &   -1.4946    \\
    1.0  ($k =    -8.9526$)  &     92.72      &  31.1530  &   32.8128  &    -1.6598    \\
    2.0  ($k =   -12.7653$)  &     96.28      &  20.3921 &   22.4897  &    -2.0976    \\
    5.0  ($k =   -19.5278$)  &     99.11      &  8.36022  &   10.5283  &   -2.1681    \\
\hline
%\end{center}
\end{tabular}
\end{center}
%\end{minipage}}
\end{table}
\end{center}
%\end{ex}
\begin{center}
\begin{table}[h]
\centering
%\resizebox{1.0\textwidth}{!}{\begin{minipage}{\textwidth}
\caption{Tracking performance using the loss $\tilde{\ell}_1$: robust optimal solution versus non-robust optimal solution (MVN) $\lambda = 0.1, \epsilon=0.01$(cases of both losses being zero excluded).}\label{table:t6}
%\tiny
\begin{center}
\begin{tabular}{cccccccc}
\hline
 \multicolumn{1}{c}{$\eta$}  &           BT          &      \multicolumn{3}{c}{ETE}              \\
\cline{3-5}
%\cline{6-8}
                             &                       &   robust   &   non-robust &   difference  \\
                             &          $\%$         &       \multicolumn{3}{c}{($*10^{-4}$)}    \\
\cline{1-5}
    0.1  ($k =    -2.2158$)  &         58.34         &   58.5587  &    58.9210   &    -0.3623     \\
    0.2  ($k =    -3.5366$)  &         62.00         &   52.3295  &    52.9541   &    -0.6246     \\
    0.5  ($k =    -6.1208$)  &         68.65         &   41.2324  &    42.3339   &    -1.1015     \\
    0.8  ($k =    -7.9434$)  &         72.98         &   34.5631  &    36.0577   &    -1.4946     \\
    1.0  ($k =    -8.9526$)  &         75.20         &   31.1530  &    32.8128   &    -1.6598     \\
    2.0  ($k =   -12.7653$)  &         82.69         &   20.3921  &    22.4897   &    -2.0976     \\
    5.0  ($k =   -19.5278$)  &         91.66         &   8.36022  &    10.5283   &    -2.1681     \\
\hline
\end{tabular}
\end{center}
%\end{minipage}}
\end{table}
\end{center}

\noindent \textbf{Example}  [Multivariate {\tt{t}} (MVT)].
 Suppose now that  the nominal distribution is a $\ell$-dimensional multivariate t distribution $t_{\ell}(\bm{\mu}_{1},\bm{\Sigma}_{1},\nu_{1})$ and an actual distribution is taken to be an
 $\ell$-dimensional multivariate t distribution $t_{\ell}(\bm{\mu}_{2},\bm{\Sigma}_{2},\nu_{2})$
 %which has a Bregman divergence of $\eta$, where $\eta = 0.1, 0.5, 1.0, 1.5, 2.0, 3.0, 4.0, 5.0$,
 such that $\bm{\mu}_{2} = k\bm{\mu}_{1}$, $\bm{\Sigma}_{2} = \bm{\Sigma}_{1}$
 %and $\nu_{2} = \nu_{1} $
 for some $k \in \mathbb{R}$.\\
%\end{ex}

First, we note that a multivariate $t$ distribution $t_{l}(\bm{\mu},\bm{\Sigma},\nu)$ has a density (see for example, \cite[p99]{NadK08}):
\begin{eqnarray}
\label{tdist}
 f(\bm{x}) & = & \frac{\Gamma(\frac{l+\nu}{2})}{(\pi\nu)^\frac{l}{2}\Gamma(\frac{\nu}{2})|\Sigma|^\frac{1}{2}}\Big(1+\frac{1}{\nu}(\bm{x}-\bm{\mu})^{\intercal}\Sigma^{-1}(\bm{x}-\bm{\mu})\Big)^{-\frac{(l+\nu)}{2}}.
\end{eqnarray}
We remind the reader that  the matrix $\Sigma $ in (\ref{tdist}) is {\it not} the covariance matrix of the multivariate $t$ distribution, but the covariance matrix is defined for every $\nu >2$ and can be expressed as $\frac{\nu}{\nu -2}\Sigma.$ \\

We again applied the smoothed loss function $\tilde{\ell}_1(.)$ from Section \ref{emrit} with $\epsilon = 0.01$ and $\lambda =0.1.$ We present results below for the case of 10 degrees of freedom but we experimented with many other values for the degrees of freedom and the results follow the same pattern. We varied the values of $k$ through the range $1, 0, -1, -2, -3, -4, -5, -8$ and calculated the resulting  radius $\eta $ numerically. The portfolio weights generally stabilized with fewer than  the 1,000,000 simulations we performed. As before, we registered the percentage of cases in which the Bregman-based portfolio was not worse than the  non-robust portfolio with respect to the loss $\ell_1(.)$ . (The non-robust portfolio was defined as minimizing the same loss $\ell_1(x)$ but without considering a neighbourhood around the nominal distribution.)
%\end{ex}
 The results are summarized in Table \ref{table:t7} and similar results  can be  obtained when $\tilde{\ell}_2(.)$ is used as a smooth loss function.\\

As before, an additional table is provided for the market downturn scenario, where we exclude ``zero'' loss results where the robust and non-robust strategies performed equally. As shown in Table \ref{table:t8}, the proportion BT was again in favor of the robust portfolio.\\

%As in the case of the multivariate normal, also here we performed the seemingly fairer comparison by excluding the cases where we ended %up with both robust and non-robust loss  $\tilde{\ell}_1(.) $ were equal to zero (i.e., we only included the cases where the robust %strategy delivered a truly better outcome). As in the normal case, the proportion BT was reduced but was still high enough to confirm the %outperformance of the robust strategy in a market downturn scenario. These results are presented in  Table \ref{table:t8}. \\

The results of this section clearly outline the significant benefits of using the robust tracking portfolio in a market downturn scenario also for the case where the nominal distribution is heavy-tailed (such as the multivariate $t$ with 10 degrees of freedom). When $k=1,$ the radius $\eta$ is zero and the robust and non-robust strategies coincide, hence, up to a negligible numerical effect, the percentage was about $75\%$ in Table \ref{table:t7} and about   $50\%$ in Table \ref{table:t8}. As $k$ starts getting smaller, the advantage of the robust approach starts popping up and is increasing monotonically when the  magnitude of $k$ increases.
\begin{center}
\begin{table}[h]
\centering
%\resizebox{1.0\textwidth}{!}{\begin{minipage}{\textwidth}
\caption{Tracking performance using the loss $\tilde{\ell}_1$: robust optimal solution versus non-robust optimal solution (MVT) $\lambda = 0.1, \epsilon=0.01, \nu=10 $ (cases where both losses are zero are included).}\label{table:t7}
%\tiny
\begin{center}
\begin{tabular}{cccccccc}
\hline
\multicolumn{1}{c}{k, $\eta$}&     BT      &         \multicolumn{3}{c}{ETE}          \\
\cline{3-5}
%\cline{6-8}
                             &             &   robust  &  non-robust  &  difference   \\
                             &    $\%$     &       \multicolumn{3}{c}{($*10^{-4}$)}   \\
\cline{1-5}
            $k =  1$         &   74.0965  &   72.9088  &   72.9100    &  -0.0012      \\
            $k =  0$         &   75.6374  &   67.4686  &   67.5311    &  -0.0625      \\
            $k = -1$         &   77.1499  &   62.0974  &   62.4280    &  -0.3306      \\
            $k = -2$         &   78.5679  &   56.8677  &   57.5992    &  -0.7315      \\
            $k = -3$         &   80.1557  &   51.8374  &   53.0417    &  -1.2043      \\
            $k = -4$         &   81.8112  &   47.0439  &   48.7512    &  -1.7073      \\
            $k = -5$         &   83.5051  &   42.5242  &   44.7248    &  -2.2006      \\
            $k=-8$           &    88.1117        &      30.7946      &     34.1600         &   -3.3655           \\
 \hline
\end{tabular}
\end{center}
%\end{minipage}}
\end{table}
\end{center}

\begin{center}
\begin{table}[h]
\centering
%\resizebox{1.0\textwidth}{!}{\begin{minipage}{\textwidth}
\caption{Tracking performance using the loss $\tilde{\ell}_1$: robust optimal solution versus non-robust optimal solution (MVT) $\lambda = 0.1, \epsilon=0.01, \nu=10 $ (cases of both losses being zero excluded).}\label{table:t8}
%\tiny
\begin{center}
\begin{tabular}{cccccccc}
\hline
\multicolumn{1}{c}{k, $\eta$}&      BT        &         \multicolumn{3}{c}{ETE}           \\
\cline{3-5}
%\cline{6-8}
                             &                &   robust  &   non-robust   &  difference  \\
                             &     $\%$       &       \multicolumn{3}{c}{($*10^{-4}$)}    \\
\cline{1-5}
            $k =  1$         &     50.77      &  72.9088  &    72.9100     &  -0.0012     \\
            $k =  0$         &     51.55      &  67.4686  &    67.5311     &  -0.0625     \\
            $k = -1$         &     52.66      &  62.0974  &    62.4280     &  -0.3306     \\
            $k=  -2$         &     53.50      &  56.8677  &    57.5992     &  -0.7315     \\
            $k=  -3$         &     54.53      &  51.8374  &    53.0417     &  -1.2043     \\
            $k = -4$         &     55.68      &  47.0439  &    48.7512     &  -1.7073     \\
            $k=  -5$         &     56.93      &  42.5242  &    44.7248     &  -2.2006     \\
              $k=-8$           &    60.72        &    30.7946        &       34.1600         &      -3.3655           \\
              \hline
\end{tabular}
\end{center}
%\end{minipage}}
\end{table}
\end{center}

\subsection{Real data example}
The simulation method offers the right vehicle to examine the merits of our methodology. This is because only via simulations (where we know the nominal, the actual  model and the true radius of contamination) we can sense the effects of our methodology. However, an indirect way to investigate these effects would be to investigate the out-of-sample performance on some real data.
We are not aware of another method in the literature for robust index tracking (in the sense we used it here as robustness
to deviations from a nominal model), with which to compare. Hence we illustrate the performance of our method on data from the Hang Seng index (Hong Kong).
This index contains 31 stocks.  The paper \cite{BeaMC03} contains many data sets that have been refereed to since then in other papers. In particular, \cite{GuaS12}, too, uses these data sets. The description of the data sets and a way to access it can be found on page 60 of \cite{GuaS12}.
 We are following the same design as done in \cite{GuaS12},  and apply the same experimental setting as described in Table 1 in their paper. For each stock, 291 consecutive weekly prices are provided. From these, the first 104 weekly prices are chosen as in-sample observations, i.e. the initial time period is $T = 104.$  The the next $52$ weeks are set to be the out-of-sample (or validation) period.
 \\

Based on the testing data, we used the MATLAB procedure {\tt{stepwiseglm}} to extract 12 stocks to include in our portfolio. This procedure creates a linear regression model
using the stepwise regression approach by adding or removing predictors from the set of 31 stocks to model the response (i.e., the index value in our case).
Admittedly, this may not be the best way to choose stocks to be included in the index tracking portfolio but it is a reasonable way of doing it and, as discussed in the introduction, we are unable to discuss rigorously the problem of selecting the ``optimal subset'' selection in this paper. As a result of applying the MATLAB  procedure, we obtained stocks with indices
$4, 11, 12, 13, 15, 18, 21, 22, 23, 25, 26, 27$ which we included in our portfolio. \\

 The Table \ref{table:t9} below gives the estimated correlations matrix of the returns of the chosen 12 stocks (columns 1 to 12) and the return of the index (column 13), rounded by two decimal places.

\begin{table}[h]
\centering
%\resizebox{1.0\textwidth}{!}
%{\begin{minipage}{\textwidth}
\caption{Estimated correlation matrix between the returns of the chosen 12 stocks (columns 1-12) and the index (column 13)}\label{table:t9}
%\tiny
\begin{center}
\begin{tabular}{ccccccccccccc}
\hline
1&          0.60 &       0.66&      0.69&        0.64& 	0.76&	0.83&	0.61&	0.31&	0.70&	0.56&	0.69&	0.82\\
0.60&	1&	0.63&	0.64&	0.51&	0.65&	0.67&	0.67&	0.32&	0.59&	0.47&	0.57&	0.77\\
0.66&	0.63&	1&	0.80&	0.67&	0.74&	0.73&	0.69&	0.35&	0.67&	0.67&	0.76&	0.86\\
0.69&	0.64&	0.80&	1&	0.70&	0.81&	0.75&	0.73&	0.33&	0.73&	0.68&	0.78&	0.89\\
0.64&	0.51&	0.67&	0.70&	1&	0.71&	0.65&	0.64&	0.37&	0.64&	0.69&	0.64&	0.82\\
0.76&	0.65&	0.74&	0.81&	0.71&	1&	0.72&	0.70&	0.27&	0.77&	0.64&	0.74&	0.88\\
0.83&	0.67&	0.73&	0.75&	0.65&	0.72&	1&	0.66&  0.27&	0.71&	0.60&	0.78&	0.87\\
0.61&	0.67&	0.69&	0.73&	0.64&	0.70&	0.66&	1&	0.35&	0.69&	0.58&	0.69&	0.84\\
0.31&	0.32&	0.35&	0.33&	0.37&	0.27&	0.27&	0.35&	1&	0.33&	0.33&	0.28&	0.41\\
0.70&	0.59&	0.67&	0.73&	0.64&	0.77&	0.71&	0.69&	0.33&	1&	0.54&	0.73&	0.83\\
0.56&	0.47&	0.67&	0.68&	0.69&	0.64&	0.60&	0.58&	0.33&	0.54&	1&	0.62&	0.76\\
0.69&	0.57&	0.76&	0.78&	0.64&	0.74&	0.78&	0.69&	0.28&	0.73&	0.62&	1&	0.86\\
0.82&	0.77&	0.86&	0.89&	0.82&	0.88&	0.87&	0.84&	0.41&	0.83&	0.76&	0.86&	1\\
\hline
\end{tabular}
\end{center}
%\end{minipage}}
\end{table}

 The correlation matrix is clearly non-diagonal. It has one outstanding large eigenvalue  (8.93689), with the  smallest eigenvalue being 0.0051 and the remaining ones are all larger than 0.0051 but smaller than one.

 Replacing expected values by empirical averages  in the main Theorem \ref{thm:bregmansystem}, we are able to find the optimal weighting for our stock selection: $$0.0841,  0.1365, 0.0429, 0.0760, 0.2040,  0.0476,  0.0528,  0.0908, 0.0253, 0.0446, 0.0888, 0.1067 $$ (rounded up to the fourth decimal). The values for the other parameters were: $\lambda=0.2, \eta=0.005 .$  These wights should be compared to the optimal non-robust weights: $$0.0786, 0.1317, 0.0487, 0.0722, 0.2026, 0.0508, 0.0567, 0.0931, 0.0250, 0.0443, 0.0902, 0.1060. $$ Following the spirit of myopic index tracking, we then applied a  moving window of size 104 when working on the
remaining out-of-sample 52 data points. Each time, a new data point was added the earliest one was deleted when re-calculating the weights in the next time point.
That is, at each time point after  $t=104,$ the sliding window contained the last 104 time points when re-calculating the weights.
The predicted weights form time point $t$ were used to update the portfolio weights to be applied at time point $t+1.$ We applied the non-robust approach, too, and compared the outcomes of the two procedures. The ETE values for the robust and non-robust approach were almost identical on the testing data set: $9.9707*10^{-6}$ for robust versus
$ 9.9552*10^{-6}$ for the non-robust but the robust approach outperformed the non-robust in the out-of-sample performance ($2.8869*10^{-5}$ for robust versus $2.9152*10^{-5}$
for non-robust).
 The  important quantity of interest, BT, was also in favour of the robust approach again in the out-of-sample performance (27 out of 52) when the quadratic loss $\ell(x)=x^2$ was used.  This increased to 42 out of 52, i.e., $80.77\%$, when the $\tilde{\ell}_1(x)$ of Equation \ref{smoothell1simple}, with $\epsilon=0.01$ was used.
These outcomes are in line with the heuristic of the methodology and with the simulations results of Sections 5.1 and  5.2. Figure 1 illustrates the  accuracy of
the approximation of the vector $B_t=1+b_t$ via $\bm{R_t}^\intercal \bm{u}^*$ across all 156 time points.  It is perhaps not surprising that we obtain such an excellent index tracking performance as the proportion of chosen stocks (12 out of 31) is relatively high, and in the same setting as in \cite{GuaS12}, two years of in-sample training and a subsequent 12 months of out-of-sample testing has been applied.

\begin{figure}
  \begin{center}
\includegraphics[width=175mm]{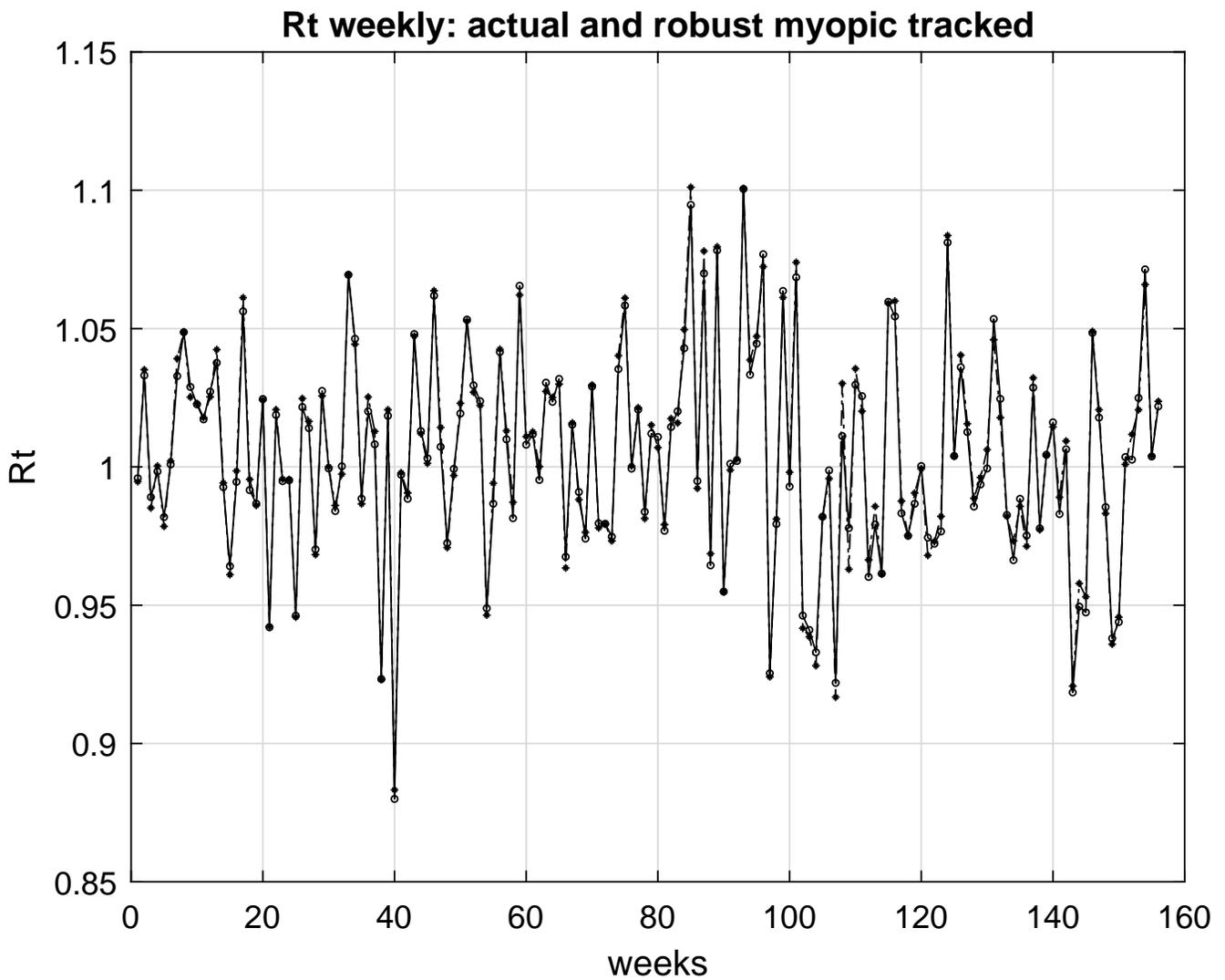}
\caption{Time series plot of $B_t=1+b_t $ and of its estimation using $\bm{R_t}^\intercal \bm{u}^*$-the robust myopic index tracking method. First 104 points: learning sample, next 52 points: out-of-sample approximation. Legend: circle-observed, asterix:robust estimator}
    \label{fig1}
  \end{center}
\end{figure}
%end red}

\newpage
\section{Discussion}
Various extensions of the suggested approach are of interest and are left for  further research. As is to be expected, the robustness effect depends on more than one factor. The value of $\lambda ,$  the ``radius of contamination'' $\eta$ , and the
model distribution, all have an effect on performance and a thorough investigation of their interplay needs to be addressed in the future. Obviously, the value of the chosen radius $\eta$ of the divergence ball  is strongly related to the amount of contamination around the nominal model. This information, especially in realistic financial portfolios, is difficult to access. However our simulations lead us to believe that even with a slightly miss-specified value of $\eta ,$ one can still enjoy the improvement delivered by the robust procedure. \\

Another adjustment parameter is the $\lambda $ value in the definition of the divergence. As pointed by  \cite{BHHJ98}, there is no universal way of selecting it.  It becomes apparent that the choices of $\eta $ and $\lambda$ must be inter-related. In simplistic situations, recommendations in this paper about the choice of $\lambda$ are  given as a way of compromise by fixing an acceptable level of efficiency loss for gaining robustness, but  a thorough study of the issue is lacking. This represents an avenue for future research.\\

Another important question is how to measure the quality of the index tracking. We have focused on the quadratic loss function since this is the typical choice in the portfolio tracking literature. This choice  equally penalizes the performance of the portfolio whenever it deviates by the same magnitude irrespectively of  whether the deviation is above of below the value of the index. A more reasonable choice of loss can be based on the functions $\ell_1$ and $\ell_2$ discussed in Section \ref{emrit}.  Using such type of  loss functions delivers a  better performance in a clear downturn market scenario as shown in Section \ref{numerics}. However, in alternative mixed scenarios, and since there is often no clear separation between market upturn and market downturn in reality,  using this loss may be disadvantageous for the investor as it may reduce their average gains. Further research in this direction will also be beneficial. A thorough comparison with the approach from \cite{Roll1992} is on our agenda. \\

Finally, the numerical examples in this paper illustrated effects when the actual distribution is on a maximal allowable distance from the nominal. This is not necessarily the least-favorable distribution: the least favorable distribution will never be known in practice and in the theoretical discussion we can only get it in the semi-closed form (\ref{semiclosed}) as a part of an implicit solution of an equation system. Despite this, the actual distributions we used for numerical illustrations  still give a good proxy of the expected effect. Of course, in our theoretical derivations, we did not need the explicit form of the least-favorable distribution and the  derivations in Section \ref{sec:semianalyopall} remain universally valid.  The mean vector and the covariance matrix used in our simulations were selected to be close to the daily returns in the Australian share market. For daily returns, assuming multivariate normality is often appropriate. However, when the returns are collected from a longer time horizon or when the nominal distribution itself is different from the multidimensional normal, the benefits of the robust portfolio  in the case of quadratic loss may be more or less spectacular depending on how heavy-tailed the nominal distribution turns out to be. In this case, explicit formulae for the divergence, such as, for example, (\ref{fullform}), would rarely be available. This does not prevent our methodology from working- the required expected values under the nominal distribution in the main Theorem \ref{thm:bregmansystem} are to be replaced with their empirical counterparts. More numerical work could demonstrate the advantages of the methodology in such heavy-tailed cases.\\

 As mentioned in the introduction, we do not discuss  the approaches to sparse index tracking portfolio selection in this paper. The combined requirement for sparsity and robustness in index tracking leads to interesting and challenging optimization problems that are left as a future research avenue.

\section{Appendix}
\subsection{
Justification of (\ref{fullform})  and (\ref{eqn:bregdivsimp})} Direct  calculation yields:
\begin{eqnarray}
\label{generalformula}
            D_{Breg}(\mathcal{E})
  &  =  &   \mathbb{E}\Big(\frac{1}{\lambda}\mathcal{E}^{\lambda+1} - \frac{\lambda+1}{\lambda}\mathcal{E} + 1\Big)                        \nonumber \\
  &  =  &   \frac{1}{\lambda}\mathbb{E}(\mathcal{E}^{\lambda+1}) - \frac{1}{\lambda}                                                       \nonumber \\
% \end{eqnarray}
% \begin{eqnarray}
  &  =  &   \frac{1}{\lambda}\mathbb{E}\Bigg(\Bigg(\frac{\sqrt{\det(2\pi\bm{\Sigma}_{1})}}{\sqrt{\det(2\pi\bm{\Sigma}_{2})}}\frac{\exp\Big(-\frac{1}{2}(\bm{X}-\bm{\mu}_{2})^{\intercal}\bm{\Sigma}^{-1}_{2}(\bm{X}-\bm{\mu}_{2})\Big)}
            {\exp\Big(-\frac{1}{2}(\bm{X}-\bm{\mu}_{1})^{\intercal}\bm{\Sigma}^{-1}_{1}(\bm{X}-\bm{\mu}_{1})\Big)}\Bigg)^{\lambda+1}\Bigg) \nonumber \\
           & &  - \frac{1}{\lambda}      \nonumber \\
  %&  =  &   %\frac{1}{\lambda}\Big(\frac{\big(\det(\bm{\Sigma}_{1}\bm{\Sigma}_{2}^{-1})\big)^{\lambda+1}\det(2\pi\tilde{\bm{\Sigma}}_{\lambda})}{\det(2\pi\bm{\Sigma}_{1})}\Big)^{\frac{1}{2}}
  %          \exp\Big(-\frac{\lambda+1}{2}\bm{\mu}_{2}^{\intercal}\bm{\Sigma}_{2}^{-1}\bm{\mu}_{2}   \nonumber \\
  %&     &   + \frac{\lambda}{2}\bm{\mu}_{1}^{\intercal}\bm{\Sigma}_{1}^{-1}\bm{\mu}_{1} %+\frac{1}{2}\tilde{\bm{\mu}}_{\lambda}^{\intercal}\tilde{\bm{\Sigma}}_{\lambda}^{-1}\tilde{\bm{\mu}}_{\lambda}\Big)\cdot
  %           \nonumber \\
  %         &     & %\int\frac{1}{\sqrt{\det(2\pi\tilde{\bm{\Sigma}}_{\lambda})}}\exp\Big(-\frac{1}{2}\big(\bm{x}-\tilde{\bm{\mu}}_{\lambda}\big)^{\intercal}\tilde{\bm{\Sigma}}_{\lambda}^{-1}
  %          \big(\bm{x}-\tilde{\bm{\mu}}_{\lambda}\big)\Big)d\bm{x}
  %        - \frac{1}{\lambda}       \nonumber \\
  &  =  &
  \frac{1}{\lambda}\Big(\frac{\big(\det(\bm{\Sigma}_{1}\bm{\Sigma}_{2}^{-1})\big)^{\lambda+1}\det(\tilde{\bm{\Sigma}}_{\lambda})}{\det(\bm{\Sigma}_{1})}\Big)^{\frac{1}{2}}
            \exp\Big(-\frac{\lambda+1}{2}\bm{\mu}_{2}^{\intercal}\bm{\Sigma}_{2}^{-1}\bm{\mu}_{2}     \nonumber \\
  &     & + \frac{\lambda}{2}\bm{\mu}_{1}^{\intercal}\bm{\Sigma}_{1}^{-1}\bm{\mu}_{1}+ \frac{1}{2}\tilde{\bm{\mu}}_{\lambda}^{\intercal}\tilde{\bm{\Sigma}}_{\lambda}^{-1}\tilde{\bm{\mu}}_{\lambda}\Big) - \frac{1}{\lambda}.             \nonumber
\end{eqnarray}
By substituting $\bm{\Sigma}_{1} = \bm{\Sigma}_{2} = \bm{\Sigma}$, we obtain (\ref{eqn:bregdivsimp}).
\subsection{Justification of the solution to the outer optimization problem }
\label{finishing}
 We notice that

  \begin{eqnarray}
              L^{inner}(\mathcal{E}^{\ast},\bm{u})
   &  =  &    \mathbb{E}\Big(H(\bm{u})\mathcal{E}^{\ast} + \alpha^{\ast}(G(\mathcal{E}^{\ast}) - \eta) + \beta^{\ast}(\mathcal{E}^{\ast} - 1)\Big)                          \nonumber \\
   &  =  &    \mathbb{E}\Bigg(H(\bm{u})\Big(\frac{\lambda}{\lambda+1}\Big(\frac{-\beta^{\ast} - H(\bm{u})}{\alpha^{\ast}}\Big) + 1\Big)^{\frac{1}{\lambda}}\Bigg)           \nonumber \\
   &     &  + \alpha^{\ast}\mathbb{E}\Bigg(\frac{1}{\lambda}\Big(\frac{\lambda}{\lambda+1}\Big(\frac{-\beta^{\ast} - H(\bm{u})}{\alpha^{\ast}}\Big) + 1\Big)^{\frac{\lambda+1}{\lambda}}
              \nonumber \\
   &     &  - \frac{\lambda+1}{\lambda}\Big(\frac{\lambda}{\lambda+1}\Big(\frac{-\beta^{\ast} - H(\bm{u})}{\alpha^{\ast}}\Big) + 1\Big)^{\frac{1}{\lambda}} + 1 - \eta\Bigg) \nonumber \\
   &     &  + \beta^{\ast}\mathbb{E}\Bigg(\Big(\frac{\lambda}{\lambda+1}\Big(\frac{-\beta^{\ast} - H(\bm{u})}{\alpha^{\ast}}\Big) + 1\Big)^{\frac{1}{\lambda}} - 1\Bigg).    \nonumber
  \end{eqnarray}

  We write the Lagrangian of the outer optimization problem.
  \begin{eqnarray}
              L^{outer}(\mathcal{E}^{\ast},\bm{u})
   &  =  &    \mathbb{E}\Bigg(H(\bm{u})\Big(\frac{\lambda}{\lambda+1}\Big(\frac{-\beta^{\ast} - H(\bm{u})}{\alpha^{\ast}}\Big) + 1\Big)^{\frac{1}{\lambda}}\Bigg)   \nonumber \\
   &     &  + \alpha^{\ast}\mathbb{E}\Bigg(\frac{1}{\lambda}\Big(\frac{\lambda}{\lambda+1}\Big(\frac{-\beta^{\ast} - H(\bm{u})}{\alpha^{\ast}}\Big) + 1\Big)^{\frac{\lambda+1}{\lambda}}  \nonumber   \\
   &     &  - \frac{\lambda+1}{\lambda}\Big(\frac{\lambda}{\lambda+1}\Big(\frac{-\beta^{\ast} - H(\bm{u})}{\alpha^{\ast}}\Big) + 1\Big)^{\frac{1}{\lambda}} + 1 - \eta\Bigg)          \nonumber \\
   &     &  + \beta^{\ast}\mathbb{E}\Bigg(\Big(\frac{\lambda}{\lambda+1}\Big(\frac{-\beta^{\ast} - H(\bm{u})}{\alpha^{\ast}}\Big) + 1\Big)^{\frac{1}{\lambda}} - 1\Bigg)   \nonumber  \\
   &     &  - \theta\Big(\bm{1}^{\intercal}\bm{u} - 1\Big).         \nonumber
  \end{eqnarray}
  The first order condition can then be obtained:
  \begin{eqnarray}\label{eqn:optimrobuststr}
              \mathbb{E}\Bigg(\frac{\partial{H}}{\partial \bm{u}}\Big(\frac{\lambda}{\lambda+1}\Big(\frac{-\beta^{\ast} - H(\bm{u})}{\alpha^{\ast}}\Big) + 1\Big)^{\frac{1}{\lambda}}\Bigg)
    &  =  &   \theta^{\ast}\bm{1},                                    \nonumber  \\
              \bm{1}^{\intercal}\bm{u}
    &  =  &   1.
  \end{eqnarray}

  To check that the solution of the above equation is optimal, we calculate the Hessian. For $\bm{y} \in \mathbb{R}^{n}$, we see that
  \begin{eqnarray}
              \bm{y}^{\intercal}\Big(\frac{\partial^{2} L^{outer}(\mathcal{E}^{\ast},\bm{u})}{\partial \bm{u}\bm{u}^{\intercal}}\Big)\bm{y}
   &  =  &    \mathbb{E}\Bigg(\bm{y}^{\intercal}\frac{\partial^{2}{H}}{\partial \bm{u}\bm{u}^{\intercal}}\bm{y}\mathcal{E}^{\ast}
            - \frac{1}{\alpha^{\ast}}\frac{1}{1+\lambda}\Big(\bm{y}^{\intercal}\frac{\partial{H}}{\partial \bm{u}}\Big)^{2}\big(\mathcal{E}^{\ast}\big)^{\frac{1}{1-\lambda}}\Bigg)    \nonumber
  \end{eqnarray}
  Since $\alpha^{\ast} > 0$,
  \begin{eqnarray}
              \bm{y}^{\intercal}\frac{\partial^{2}{H}}{\partial \bm{u}\bm{u}^{\intercal}}\bm{y}
     &  =  &  -2\big(\bm{y}^{\intercal}\bm{R}\big)^{2}   \leq  0,      \nonumber
  \end{eqnarray}
  it is easy to see that the Hessian is negative semi-definite. As a consequence, we have verified that the solution of (\ref{eqn:optimrobuststr}) is optimal, and we will denote it as $\bm{u}^{\ast}$. This finishes the proof of Theorem \ref{thm:bregmansystem}. \\

\section*{Funding} This work was supported by the Australian Research Council's Discovery Project funding scheme (Project DP160103489).


\begin{thebibliography}{99}

 \bibitem[{Amari (2016)}]{Amari16} Amari,Shun-ichi, 2016.  Information Geometry and Its Applications. Springer, Japan.
\bibitem[{Amari \& Cichocki (2010)}]{AmariCich2010} Amari, S. \& Cichotski, A. (2010). Information geometry of divergence functions.  Bulletin of the Polish academy of scinces. Tehcnical sciences, 58~(1), 183--195.

\bibitem[{Andriosopoulos and Nomikos(2014)}]{AndN14}
Andriosopoulos, K. \& Nomikos, N. (2014). Performance replication of the Spot Energy index with optimal equity portfolio selection: Evidence from the UK, US and Brazilian markets.
European Journal of Operational Research 234~(2), 571--582.

\bibitem[{Banerjee et~al. (2005)}]{BMDG05}  Banerjee, A.\& Merugu, S. \& Dhillon, I. \& Ghosh, J. (2005). Clustering with Bregman Divergences. Journal of Machine Learning Research 6, 1705--1749.


\bibitem[{Basu et~al. (1998)}]{BHHJ98} Basu, A., \& Harris, N., \& Hjort, N., \& Jones, M.C. (1998). Robust and efficient estimation by minimizeing a denisty power divergence. Biometrika 85 , 549--559.

\bibitem[{Beasley et~al.(2003) Beasley and Meade and Chang}]{BeaMC03}
Beasley J. E., A. \& Meade, N. \& Chang T.J. (2003). An evolutionary heuristic for the index tracking problem.
European Journal of Operational Research 148~(3),  621--643.

 \bibitem[{Benidis et~al. (2018)}]{Benidis18} Benidis, K., \& Feng, Y., \& Palomar, D. (2018) Sparse Portfolios for High-Dimensional Financial Index Tracking. IEEE Transactions on Signal Processing 66~(1), 155-170.

\bibitem[{Ben-Tal et~al.(1988)}]{BenTC88}
Ben-Tal, A. \& Teboulle, M. \& Charnes A. (1988). The role of duality in optimization problems involving entropy functionals with applications to information theory.
Journal of Optimization Theory and Applications 58~(2), 209--223.

\bibitem[{Blanchet et al. (2019)}]{BLTY2019}  Blanchet, J., Lam, H., Tang, Q. and Yuan, Z. (2019). Robust actuarial risk analysis. North American Actuarial Journal 23~(1), 33-63.


\bibitem[{Boyd and Vandenberghe(2004)}]{BoyV09}
Boyd, S. \& Vandenberghe, L. (2004). Convex Optimization. (Seventh printing with corrections 2009)
Cambridge University Press

\bibitem[{Bregman(1967)}]{Bre67}
Bregman, L. M. (1967). The relaxation method of finding the common points of convex sets and its application to the solution of problems in convex programming.
USSR Computational Mathematics and Mathematical Physics 7~(3), 200--217.

 \bibitem[{Canakgoz and Beasley (2008)}]{Canakgoz2008}  Canakgoz, N. A.\& Beasley, J. E. (2009). Mixed-integer programming approaches
for index tracking and enhanced indexation. European Journal of Operational Research 196 (1), 384--399.

\bibitem[{Chen and Mangasarian (1995)}]{CM95} Chen, C. and Mangasarian, O. (1995). Smoothing methods for convex inequalities and linear complementarity problems. Mathematical Programming 71, 51--69.

\bibitem[{Chiam et~al.(2013) Chiam and Tan and Mamun}]{ChiTM13}
Chiam S. C., A. \& Tan, K.C. \& Mamun A.A. (2013). Dynamic index tracking via multi-objective evolutionary algorithm.
Applied Soft Computing 13~(7),  3392--3408.

\bibitem[{Cichocki and Amari(2010)}]{CicA10}
Cichocki, A. \& Amari, S. (2018). Families of Alpha- Beta- and Gamma- Divergences: Flexible and Robust Measures and Similarities.
Entropy 12, 1532--1568.

\bibitem[{Dey and Juneja(2010)}]{DeyJ10}
Dey, S. \& Juneja, S. (2010). Entropy approach to incorporate fat tailed constraints in financial models. SSRN electronic journal.
https://papers.ssrn.com/sol3/papers.cfm?abstract\_id=1647048, Tata Institute of Fundamental Research Mumbai, India (last download 27 11 2017).

%\bibitem[{Gn\"agi and Strub (2020)}]{GnStr2020} {\color{red} Gn\"agi, M.\& Strub, O. (2020). Tracking and outperforming large stockmarket indices. Omega 90, 101999.}

%\bibitem[{Guastaroba and Speranza (2012)}]{GuastSpe2012} {\color{red} Guastaroba, G. \& Speranza, M.G. (2012). Kernel Search: An application to the index tracking problem. %European Journal of Operational Research
%217, 54--68.}

\bibitem[{de Paulo et~al.(2010)}]{dePdC16}
de Paulo, W. L. \& de Oliveira, E. M. \& Vosta, O. L. do V. (2016). Enhanced index tracking optimal portfolio selection.
Finance Research Letter 16, 92--102.

\bibitem[{Dose and Cincotti(2005)}]{DosC05}
Dose, C. \& Cincotti, S. (2005). Clustering of financial time series with application to index and enhanced index tracking portfolio.
Physica A: Statistical Mechanics and its Applications 355~(1), 145--151.

 \bibitem[{Fabozzi (2007)}]{Fabozzi2007}   Fabozzi, F. J. (2007). Robust portfolio optimization and management. Wiley, New Jersey.

\bibitem[{Filippi et~al.(2016)}]{Filippi2016} Filippi, C. \&  Guastaroba, G. \&  Speranza, M. (2016). A heuristic framework
for the bi-objective enhanced index tracking problem. Omega 65, 122--137.

\bibitem[{Gaivoronski et~al.(2005)}]{GaiKW05}
Gaivoronski, A. A. \& Krylov, A. \& Wijst, N. van der (2005). Optimal portfolio selection and dynamic benchmark tracking.
European Journal of Operational Research 163~(1), 115--131.

\bibitem[{Glasserman and Xu(2014)}]{GlaX14}
Glasserman, P. \& Xu, X. B. (2014). Robust risk measurement and model risk.
Quantitative Finance 14~(1), 29--58.


\bibitem[{Gn\"agi and Strub (2020)}]{Gnagi2020} Gn\"agi,M. \& Strub,  O. (2020). Tracking and outperforming large stock-market
indices. Omega, 90, 101999.


\bibitem[{Goh and Dey (2014)}]{GD14}
Goh, G. \& Dey, D. (2014). Bayesian model diagnostics using functional Bregman divergence. Journal of Multivariate Analysis 124, 371--383.

\bibitem[{Guastaroba et~al.(2016) Guastaroba and Mansini and Ogryczak and Speranza}]{GuaMOS16}
Guastaroba, G. \& Mansini, R. \& Ogryczak W. \& Speranza, M.G. (2016). Linear programming models based on Omega ratio for the enhanced index tracking problem.
European Journal of Operational Research 251~(3), 938--956.

\bibitem[{Guastaroba and Speranza(2012)}]{GuaS12}
 Guastaroba, G. \& Speranza, M.G. (2012). Kernel Search: An application to the index tracking problem.
European Journal of Operational Research 217~(1), 54--68.

\bibitem[{Huber and Ronchetti (2009)}]{HubRon2009} Huber, P.\& Ronchetti, E. (2009). Robust Statistics, 2nd Edition. Wiley, New York.

\bibitem[{Lejeune(2012)}]{Lej12}
Lejeune, M. A. (2012). Game theoretical approach for reliable enhanced indexation.
Decision Analysis 9~(2), 146--155.

\bibitem[{Maginn et~al.(2007)}]{MagTMP07}
Maginn, J. L. \& Tuttle, D. L. \& McLeavey, D. W. \&  Pinto, J. E. (2007). Managing Investment Portfolios. (3rd Edition)
John Wiley \& Sons, Inc.

\bibitem[{Meade and Salkin(1990)}]{MeaS90}
Meade, N. \& Salkin, G. R.(1990). Developing and Maintaining an Equity Index Fund.
The Journal of Operational Research Society 41~(7), 599--607.


\bibitem[Mihoko \& Eguchi (2002)]{MiEguchi2002}Mihoko, M. \& Eguchi, S. (2002). Robust Blind Source Separation by Beta Divergence.   Neural Computation, 14, 8, 1859--1886.


\bibitem[{Montfort et~al.(2008)}]{MonVFD08}
Montfort, K. V. \& Visser, E. \& Draat, L. F. V.(2008). Index Tracking by Means of Optimized Sampling.
The Journal of Portfolio Management Winter 34~(2), 143--152.

\bibitem[{Nadarajah and Kotz(2008)}]{NadK08}
Nadarajah, S. \& Kotz, S. (2008). Estimation Methods for the Multivariate t Distribution.
Acta Applicandae Mathematicae 102~(1), 99--118.

\bibitem[{Penev and Naito(2018)}]{PenN18}
Penev, S. \& Naito, K. (2018). Locally robust methods and near-parametric asymptotics. Journal of Multivariate Analysis,  167,  395--417

\bibitem[{Penev and Prvan (2016)}]{PenPr16}
Penev, S. \& Prvan, T. (2016). Robust estimation in structural equation models using Bregman and other divergences with t-centre approach to estimate the covariance matrix. ANZIAM Journal, 56, (Proceedings CTAC2014), C339-C354.

\bibitem[{Poczos and Schneider(2011)}]{PocS11}
Poczos B., A. \& Schneider, J. (2011). On the estimation of $\alpha$-divergences. Journal of Maching Learning Research: Workshops and Conferences 15,  609--617.


 \bibitem[{Roll (1992)}]{Roll1992}   Roll, R. (1992). A mean/variance analysis of tracking error. The Journal
of Portfolio Management 18, 13--22.

\bibitem[{Strub and Baumann (2018)}]{StrB18}
Strub, O. \& Baumann, P. (2018). Optimal construction and rebalancing of index-tracking portfolios.
European Journal of Operational Research 264~(1), 370--387.

\bibitem[{Vemuri et~al. (2011)}]{VLAN11}
Vemuri, B. \& Liu, M. \& Amari, S-I. \& Nielsen, F. (2011). Total Bregman Divergence and Its Applications to DTI Analysis. IEEE Transactions on Medical Imaging, 30, 475-483.

\end{thebibliography}
\end{document}